\documentclass{IEEEtran}  

\usepackage{bbm}
\usepackage{mathrsfs}
\usepackage{verbatim}
\usepackage{amsmath}
\usepackage{amsfonts}
\usepackage{bm}
\usepackage{graphicx}
\usepackage{float}

\usepackage{amsthm}
\usepackage{xcolor}
\usepackage{mathtools}
\usepackage{cite}
\usepackage{multirow}

\usepackage[color = green!130!orange!10!]{todonotes}

\theoremstyle{definition}
\newtheorem{assumption}{Assumption}
\newtheorem{definition}{Definition}
\newtheorem{remark}{Remark}

\theoremstyle{definition}

\theoremstyle{definition}
\newtheorem{problem}{Problem}

\theoremstyle{plain}

\newtheorem{theorem}{Theorem}
\newtheorem{lemma}{Lemma}

\title{
Social Media and Misleading Information in a Democracy: A Mechanism Design Approach
}
    
\author{Aditya Dave, \textit{Student Member, IEEE}, Ioannis Vasileios Chremos, \textit{Student Member, IEEE}, \\ Andreas A. Malikopoulos, \textit{Senior Member, IEEE}
\thanks{This research was supported by the Sociotechnical Systems Center (SSC) at the University of Delaware.}%
\thanks{The authors are with the Department of Mechanical Engineering, University of Delaware, Newark, DE, 19716, USA (emails: \texttt{adidave@udel.edu}; \texttt{ichremos@udel.edu}; \texttt{andreas@udel.edu}).}%
}

\begin{document}

\maketitle

\thispagestyle{empty}
\pagestyle{empty}

\begin{abstract}

In this paper, we present a resource allocation mechanism for the problem of incentivizing filtering among a finite number of strategic social media {platforms}. We consider the presence of a strategic government and private knowledge of how misinformation affects the users of the social media platforms. Our proposed mechanism incentivizes social media {platforms} to filter misleading information efficiently, and thus indirectly prevents the spread of fake news. In particular, we design an economically inspired mechanism that strongly implements all generalized Nash equilibria for efficient filtering of misleading information in the induced game. We show that our mechanism is individually rational, budget balanced, while it has at least one equilibrium. Finally, we show that \textcolor{black}{for quasi-concave utilities and constraints, our mechanism admits a generalized Nash equilibrium and implements a Pareto efficient solution.}

%

\end{abstract}



\section{Introduction} \label{section:intro}

For the last few years, political commentators have been indicating that we live in a \textit{post-truth} era \cite{davies2016age}, wherein the deluge of information available on the internet has made it extremely difficult to identify facts. As a result, individuals have developed a tendency to form their opinions based on the \textit{believability} of presented information rather than its truthfulness \cite{cone2019believability}. \textcolor{black}{This phenomenon is exacerbated by the business practices of social media platforms, which often seek to maximize the \textit{engagement} of their users at all costs. In fact, the algorithms developed by platforms for this purpose often promote conspiracy theories among their users \cite{tufekci2018youtube}.}

The sensitivity of users of social media platforms to conspiratorial ideas makes them an ideal terrain to conduct political misinformation campaigns \cite{kramer2014experimental, weedon2017information}. Such campaigns are especially effective tools to disrupt democratic institutions, because
\textcolor{black}{the functioning of stable democracies relies on \textit{common knowledge} about the political actors and the processes they can use to gain public support \cite{farrell2018common}. The trust held by the citizens of a democracy on common knowledge includes: (i) trust that all political actors act in good faith when contesting for power, (ii) trust that elections lead to a free and fair transfer of power between the political actors, and (iii) trust that democratic institutions ensure that elected officials wield their power in the best interest of the citizens.} In contrast, citizens of democracies often have a \textit{contested knowledge} regarding who should hold power and \textcolor{black}{how they should use it \cite{farrell2018common}.}
The introduction of \textit{alternative facts} can reduce the trust on common knowledge about democracy, \textcolor{black}{especially if they become accepted beliefs among the citizens.} Such disruptions on \textcolor{black}{the trust on common knowledge} can be found in the $2016$ U.S. elections \cite{allcott2017social} and Brexit Campaign in $2016$ \cite{oxford2018russia}, \textcolor{black}{where the spread of misinformation through social media platforms} resulted in a large number of citizens mistrusting the results of voting.

To tackle this growing phenomenon of misinformation, in this paper, we consider a finite group of social media platforms, \textcolor{black}{whose users represent the citizens in a democracy,} and a democratic government. Every post in the platforms is associated with a parameter that captures its informativeness, \textcolor{black}{which can take values between two extremes: (i) completely factual and (ii) complete misinformation.}
In our framework, posts that exhibit misinformation can lead to a decrease in trust on common knowledge among the users \textcolor{black}{\cite{bessi2015, brown2018,tucker2017, sternisko2020dark}.}
In addition, social media platforms are considered to have the technologies to \textit{filter}, or label, posts that intend to sacrifice trust on common knowledge. \textcolor{black}{Thus, the government seeks to incentivize the social media platforms to use these technologies and filter any misinformation included in the posts.}

Motivated by capitalistic values, \textcolor{black}{we induce a \textit{misinformation filtering game} to describe the interactions between the social media platforms and the government. In this game, each platform acts as strategic player seeking to maximize their advertisement revenue from the engagement of their users \cite{allcott2017social, jaakonmaki2017}.
User engagement is a metric that can be used to quantify the interaction of users with a platform, and subsequently, how much time they spend on the platform.
Recent efforts reported in the literature on misinformation in social media platforms have indicated that increasing filtering of misinformation leads to decreasing of user engagement \cite{candogan2020optimal}. There are many possible reasons for this phenomenon. First, filtering reduces the total number of posts propagating across the social network. Second, the users whose opinions are filtered may perceive this action as dictatorial censorship \cite{pew2020}, and as a result, they may chose to express their opinions in other platforms. Finally, misinformation tends to elicit stronger reactions, e.g., surprise, joy, sadness, as compared to factual posts \cite{vosoughi2018spread}, which may increase user engagement.
Thus, each platform is reluctant to filter misinformation.}

\textcolor{black}{In our framework, we consider that the government is also a strategic player, whose utility increases as the trust of the users of social media platforms on common knowledge increases. Consequently, increasing filtering of misinformation by the social media platforms increases the utility of the government. Thus the government is willing to make an investment to incentivize the social media platforms to filter misinformation. In our approach, we use mechanism design to distribute this investment among the platforms optimally, and in return, implement an optimal level of filtering.}

Mechanism design was developed for the implementation of system-wide optimal solutions to problems involving multiple rational players with conflicting interests, each with private information about preferences \cite{mas_colell1995}. Note that this approach is different from traditional approaches to decentralized control with private information \cite{2019Aditya_arXiv, mahajan2012, 17, Malikopoulos2018} because the players are not a part of the same time, but in fact, have private and competitive utilities.
The fact that Mechanism design optimizes the behaviour of competing players has led to broad applications spanning different fields including economics, politics, wireless networks, social networks, internet advertising, spectrum and bandwidth trading, logistics, supply chain, management, grid computing, and resource allocation problems in decentralized systems \cite{sharma2012local,sinha2013,kakhbod2011efficient, jain2010,zhang2019efficient,chremos2020_TCNS, chremos2019social}.

The contribution of this paper is as follows. We present an indirect mechanism to incentivize social media platforms to filter misleading information.
We show that our proposed mechanism is (i) feasible, (ii) budget balanced, (iii) individual rational, and (iv) strongly implementable at the equilibria of the induced game. We prove the existence of at least one generalized Nash equilibrium and show that our mechanism induces a Pareto efficient equilibrium.

The rest of the paper is organized as follows. In Section \ref{section:formulation}, we provide the modeling framework and problem formulation. In Section \ref{section:md_problem}, we present our mechanism, and in Section \ref{section:properties_of_mechanism}, we prove the associated properties of the mechanism. In Section V, we interpret the mechanism and present a descriptive example. Finally, in Section VI we conclude and present some directions for future research.

\section{Problem Formulation} \label{section:formulation}


We consider a democratic society consisting of a finite and nonempty set of social media platforms $\mathcal{I} = \{1, \dots, I\}$, $I \in \mathbb{N},$ and a government. 
\color{black}
We refer to the social media platforms and the government collectively as the \textit{players}, and denote the set of all players by $\mathcal{J} = \mathcal{I} \cup \{0\}$, where the index $0$ corresponds to the government. The players strategically take actions in a \textit{misinformation filtering game} that is described in this section.

\subsection{Misinformation Filtering Game for Platforms}

\color{black}
Let the informativeness of a post on platform $i\in\mathcal{I}$ be denoted by $x_i \in [0,1]$, where $x_i = 0$ indicates that the post contains complete misinformation and $x_i = 1$ indicates that the post contains completely factual information. Our hypothesis, \textcolor{black}{inspired by \cite{farrell2018common, bessi2015, brown2018,tucker2017, sternisko2020dark},} states that the emergence of posts with many falsehoods and a low informativeness, i.e., $x_i \to 0$, leads to a decrease of trust of the users on common knowledge about democracy. \textcolor{black}{Recall that common knowledge about democracy refers to knowledge of political actors in a democratic society and the process they use to gain public support.}

Each social media platform $i \in \mathcal{I}$ has the technological {means} to detect and filter misinformation. \textcolor{black}{In the misinformation filtering game that we impose in our framework, the action $a_i$ of platform $i$ represents the level of filtering imposed by $i$ and takes values in a feasible set of actions $\mathcal{A} = [0,1]$. Each action $a_i$  minimizes the spread of a post that has informativeness $x_i < a_i$, while posts with informativeness $x_i \geq a_i$ are unaffected.} In practice, filtering of misinformation can be implemented in many ways.
The social media platform can place warnings on each post with $x_i < a_i$ to inform the users of their falsehood, or they can modify their algorithms to limit the propagation of such posts among users. Thus, the action $a_i$ represents the lower threshold on informativeness that is acceptable by platform $i$. \textcolor{black}{To this end, we refer to the action $a_i$ as the filter of platform $i$.}

\textcolor{black}{Each platform $i \in \mathcal{I}$ generates revenue by monetizing the \textit{engagement} of their users through advertisements \cite{jaakonmaki2017}.
By increasing filtering of misinformation there is a decrease in user engagement \cite{candogan2020optimal}.
This may be due to a perception of censorship among users \cite{pew2020}, and as a result, they may chose to express their opinions in other platforms.
Consider, for example, platform $l \in \mathcal{I}$ with a filter $a_l > a_i$. Some of the users of $l$, whose posts have been marked up by the filter, may migrate to platform $i$ which will lead to an increase in the engagement of platform $i$. This phenomenon motivates us to define a set of \textit{competing platforms}.}



\begin{definition}
For each platform $i \in \mathcal{I}$, the set $\mathcal{C}_i \subset \mathcal{I}$, with $i \in\mathcal{C}_i$, is the set of \textit{competing platforms} whose choice of filters has an impact on the engagement of platform $i$.
\end{definition}

To simplify the presentation of our results, we consider that for any two platforms $i,k \in \mathcal{I}$, if $i \in \mathcal{C}_k$, then $k \in \mathcal{C}_i$. However, our mechanism can easily be extended to the case of asymmetric competition among social media. Given the set of competing platforms $\mathcal{C}_i$, we can define a \textit{valuation function} of platform $i$.

\begin{definition}
    The \textit{valuation function} of a social media platform $i \in \mathcal{I}$ is 
    $v_i\big(a_k : k \in \mathcal{C}_i \big) : \mathcal{A}^{|\mathcal{C}_i|} \to \mathbb{R}_{\geq 0}$. \textcolor{black}{It is a decreasing function with respect to $a_i$ and strictly increasing with respect to $a_l$ for all $l \in \mathcal{C}_{-i}$, where $\mathcal{C}_{-i} = \mathcal{C}_{i} \setminus \{i\}$.}
\end{definition}

\textcolor{black}{The valuation function $v_i\big(a_k : k \in \mathcal{C}_i \big)$ corresponds to the revenue generated by platform $i$ given the user engagement after all platforms have implemented their filters. A higher value of $a_i$ will result in decreasing user engagement in platform $i$, and thus their revenue. On the other hand, a higher value of $a_l$ of another competing platform $l \in \mathcal{C}_{-i}$ will result in increasing user engagement, and thus revenue, in platform $i$.}

\textcolor{black}{Next, recall from the discussion in the previous section that filtering of misinformation in a social media platform increases the trust of the users of this platform on common knowledge about democracy. Next, for each platform $i \in \mathcal{I}$, we define the average trust function on common knowledge.}

\begin{definition}
The \textit{average trust function} of the users of platform $i \in \mathcal{I}$ on common knowledge is $h_i(a_i) : \mathcal{A} \to [0, 1]$, \textcolor{black}{and it is a strictly increasing function with respect to $a_i$.}
\end{definition}

\textcolor{black}{The average trust function $h_i(a_i)$ captures the impact of filter $a_i$ on the trust on common knowledge across the users of platform $i$.
A low value of $h_i(a_i)$ implies that $a_i$ leads to low trust on common knowledge for the users of platform $i$, and vice versa.}
In practice, platform $i$ can measure the opinions expressed by their users \cite{ceron2014every} through surveys, and over time, use these measurements to estimate the impact of filter $a_i$ using the average trust function $h_i(a_i)$.




\color{black}

\subsection{Misinformation Filtering Game for the Government}

Recall that, in our framework, the government is considered the strategic player $0 \in \mathcal{J}$. The government's objective is to maximize the trust of the users of all social media platforms on common knowledge.
Therefore, the government selects an action  $a_0\in\mathcal{A} = [0, 1]$ that designates a lower bound which must be satisfied by the aggregate average trust of all social media platforms in $\mathcal{I}$. To this end, we refer to the action $a_0$ as the government's lower bound on trust on common knowledge.

\color{black}

Let $N_i \in \mathbb{N}$ be the total number of users of the social media platform $i \in \mathcal{I}$. Then, the fraction of the number of users of $i$ with respect to the total number of users of all {platforms} is
\begin{equation}\label{eqn:n_i}
    n_i = \frac{N_i}{\sum_{l \in \mathcal{I}} N_l}.
\end{equation}
The fraction $n_i$ represents the contribution of users in platform $i$ on the average trust on common knowledge about democracy. \textcolor{black}{Since $\sum_{i \in \mathcal{I}} n_i = 1$, the aggregate average trust on common knowledge is $\sum_{i \in \mathcal{I}} n_i \cdot h_i(a_i)$.
In our framework, the government's role is to select the lower-bound $a_0$ for the aggregate average trust.}
After the government decides on $a_0$, each {platform} $i \in \mathcal{I}$ who decides to participate in the game must select a filter $a_i$ that satisfies the following constraint:
\begin{equation}\label{eqn:a_0-constraint}
    a_0 - \sum_{i \in \mathcal{I}} n_i \cdot h_i(a_i) \leq 0.
\end{equation}
Next, we define the government's valuation as a function of the lower bound on trust $a_0$.

\color{black}
\begin{definition}
The \textit{valuation function} of the government is $v_0(a_0) : [0, 1] \to \mathbb{R}_{\geq 0}$, and it is an increasing function with respect to the lower bound $a_0$.
\end{definition}
The government's valuation function $v_0(a_0)$ assigns a monetary value on the lower bound $a_0$. Recall that the government seeks to increase the trust on common knowledge among the users of all social media platforms. Thus, the government's valuation increases as the lower bound on aggregate average trust increases.
\color{black}
We also consider that the government might have limited resources available to invest in this problem, i.e., there exists a finite monetary budget $b_0 \in \mathbb{R}_{\geq 0}$ 
\color{black}
representing the maximum possible expenditure of the government for this problem.

\subsection{Information Structure}

\color{black}

\textcolor{black}{In this subsection, we specify the private and public  information structure corresponding to each player in the imposed game.}


\textcolor{black}{\textit{1) Public information:}} The set of competing platforms $\mathcal{C}_i$ and fraction of users $n_i$ of each platform $i \in \mathcal{I}$ are known to all players in set $\mathcal{J}$. \textcolor{black}{Moreover, the set of feasible actions $\mathcal{A}$ is known to all players in the set $\mathcal{J}$.}

\textcolor{black}{\textit{2) Valuation functions:}} The valuation function $v_i(\cdot)$ of each social media platform $i \in \mathcal{I}$ is considered private information, and thus, it is known only to platform $i$. Similarly, the valuation function $v_0(\cdot)$ and the budget $b_0$ of the government are private information of the government. 

\textcolor{black}{\textit{3) Average trust functions:}} The average trust function $h_i(\cdot)$ of social media platform $i \in \mathcal{I}$ is considered private information, and thus, it is known only to platform $i$ (it is not known to the government).

\color{black}
\subsection{Assumptions}

\color{black}

In the modeling framework presented above, we impose the following assumptions:

\begin{assumption} \label{assumption:cardinality}
    For each platform $i \in \mathcal{I}$, \ $|\mathcal{C}_i| \geq 3$.
\end{assumption}

We impose this assumption to simplify the exposition of our mechanism. Assumption \ref{assumption:cardinality} implies that each user subscribes in multiple social media platforms. 
\textcolor{black}{It has been shown in the literature that each user, on an average, subscribes to 8 social media platforms \cite{datareportal2020}.
Nevertheless,} we present an extension of our mechanism for $|\mathcal{C}_{i}| \geq 2$ in Appendix A.

\color{black}
\begin{assumption}\label{assumption:filter_compatibility}
    The valuation function $v_i\big(a_k : k \in \mathcal{C}_i \big) : \mathcal{A}^{|\mathcal{C}_i|} \to \mathbb{R}_{\geq 0}$ of each social media platform $i \in \mathcal{I}$ is a concave and differentiable function with respect to $a_k$.
\end{assumption}

\color{black}

The concavity of $v_i\big(a_k : k \in \mathcal{C}_i \big)$ captures the diminishing marginal change in engagement due to additional filtering. \textcolor{black}{Practically, the higher the value of $a_i$, the more users of platform $i$ will perceive the filter as censorship of their opinions.
Thus, for platform $i$, increasing a low-value filter may lead to a lesser loss in engagement as compared to increasing a filter whose value is already high.}
\textcolor{black}{Nevertheless, to ensure the robustness of our proposed mechanism,} we also present an analysis of our system by relaxing Assumption \ref{assumption:filter_compatibility} in Section IV-A.

\color{black}

\begin{assumption}\label{assumption:average_trust}
    The average trust function $h_i(a_i) : \mathcal{A} \to [0, 1]$ of each social media platform $i \in \mathcal{I}$ is a concave and differentiable function with respect to $a_i$.
\end{assumption}

\color{black}

\textcolor{black}{The concavity of $h_i(a_i)$ implies that, for large values of $a_i$, a small incremental change of $a_i$ would not have a significant impact on the average trust of users on common knowledge. Practically, this implies  low values of $a_i$ will have a major impact on the average trust function. Nevertheless, to ensure the robustness of our mechanism, we also present an analysis of our system by relaxing Assumption \ref{assumption:average_trust} in Section IV-A.}

\color{black}

\begin{assumption}\label{assumption:government_valuation}
    The valuation function of the government $v_0(a_0) : [0, 1] \to \mathbb{R}_{\geq 0}$ is a concave and differentiable function with respect to the lower-bound $a_0$.
\end{assumption}

\color{black}

\textcolor{black}{Practically, for high values of $a_0$, the government might not be interested in investing additional resources to increase $a_0$ even more, as the impact on improving common knowledge would not be significant. Nevertheless,}
we also present an analysis of our system by relaxing Assumption \ref{assumption:government_valuation} in Section IV-A.

\begin{assumption} \label{assumption:assumed_knowledge_b}
    The output of the function $h_i(a_i)$ can be monitored by any competing platform $l\in\mathcal{C}_{- i}$, and a violation of the condition \eqref{eqn:a_0-constraint} can be detected by the government.
\end{assumption}

\textcolor{black}{Assumption \ref{assumption:assumed_knowledge_b} helps us enforce the mechanism, which we present in Section \ref{section:md_problem}, in a static environment. 
In the mechanism, each platform $i \in \mathcal{I}$ commits to a minimum value of the average trust function among their users which can be achieved by choosing an appropriate value for $a_i$. Consider that a platform $i$ selects a value $a_i$ that fails to satisfy this commitment. Practically, the government can detect a violation of \eqref{eqn:a_0-constraint} by gauging public opinion on the internet and through surveys. However, the government does not know the function $h_i(\cdot)$ of platform $i$, and thus, would penalize each platform in $\mathcal{I}$ equally for the violation of \eqref{eqn:a_0-constraint}. To avoid the penalty for the failure of platform $i$, a competing platform $l \in \mathcal{C}_{-i}$ can report the violation of $i$. 
Thus, it is reasonable to consider that each platform $i \in \mathcal{I}$ monitors the output $h_l(a_l)$ of each competing platform $l \in \mathcal{C}_{-i}$ to maximize their own utility.}
We believe that using a dynamic mechanism, we could potentially relax Assumption \ref{assumption:assumed_knowledge_b} \cite{zhang2019efficient}. This would be a potential direction for future research.

\begin{assumption} \label{assumption:excludability}
    The government ensures that any social media platform $i \in \mathcal{I}$ that does not participate in the mechanism receives no benefits from the filters of participating social media.
\end{assumption}

In static mechanisms, the ability to exclude a player from receiving benefits of some common resource is a necessary condition for voluntary participation of players without any monetary investment \cite{saijo2010fundamental}. This condition is often assumed implicitly in the literature \cite{sharma2012local,sinha2013,jain2010,kakhbod2011efficient}. In our mechanism, the government can make an investment up to the budget $b_0$. \textcolor{black}{Thus, we assume \textit{partial excludability} in Assumption \ref{assumption:excludability}, where a non-participating platform $i$ still receives the maximum valuation for selecting filter $a_i = 0$, but cannot receive benefits from the filters of any participating platforms.
In practice, the government can publicize that platform $i$ has chosen not to contribute in a collective endeavor to filter misinformation.
The resulting loss in credibility among the users of the platforms that participate will minimize their migration to platform $i$. This assumption may be relaxed using a dynamic mechanism, which could be another direction for future research \cite{farhadi2019}.}

\subsection{Problem Statement}

\textcolor{black}{Since there is a conflict of interest between the government and the social media platforms, we consider that the government} hires a social planner to design a mechanism \textcolor{black}{to impose the misinformation filtering game. The mechanism must serve two purposes: (i) incentivize all platforms to voluntarily participate in the game, and (ii) induce a selection of filters that maximizes the \textit{social welfare} of the system. The social welfare of the system is the sum of utilities of all players, formally defined next.
To meet these objectives, the social planner asks each player $i \in \mathcal{J}$ to send a message $m_i$ from a set of feasible messages $\mathcal{M}_i$. Based on the message profile $m = (m_0, m_1, \dots, m_{|\mathcal{I}|})$, the social planner assigns a tax $\textcolor{black}{\tau}_i(m) \in \mathbb{R}$ for each social media platform $i \in \mathcal{I}$, and an investment $\textcolor{black}{\tau}_0(m) \in \mathbb{R}_{\geq 0}$ for the government. The message and tax of each player is formally defined in Section III-B. By convention, a tax $\textcolor{black}{\tau}_i(m) > 0$ is a payment made by player $i \in \mathcal{J}$, and a tax $\textcolor{black}{\tau}_i(m) < 0$ is a subsidy given to player $i$.}
Thus, the taxes of the platforms can be either monetary payments or subsidies, whereas, the government may never collect a monetary subsidy from any platform.
\color{black}
Note that the social planner must not receive any profit, nor incur any losses, for designing and implementing the mechanism, which implies that the mechanism should be budget balanced, i.e., $\sum_{i \in \mathcal{J}} \textcolor{black}{\tau}_i(m) = 0$. 

Next, we define the utilities of the players.

\begin{definition}
The \textit{utility} of each platform $i \in \mathcal{I}$ is given by $u_i\big(m , a_k: k \in \mathcal{C}_i\big) := v_i\big(a_k: k \in \mathcal{C}_i\big) - \textcolor{black}{\tau}_i(m),$
while the utility of the government is given by  $u_0(m, a_0) := v_0(a_0) - \textcolor{black}{\tau}_0(m).$
\end{definition}

The social welfare is $u_0(m,a_0) + \sum_{i \in \mathcal{I}} u_i(m, a_k:k\in \mathcal{C}_i)$. The optimization problem for the social planner is to maximize the social welfare, and it is formulated as follows.

\begin{problem}\label{problem1}
        \begin{align}
            \max_{a, \textcolor{black}{\tau}(m)} \bigg(v_0(a_0) &- \textcolor{black}{\tau}_0(m) + \sum_{i \in \mathcal{I}} \Big(v_i\big(a_k:k \in \mathcal{C}_i\big) - \textcolor{black}{\tau}_i(m)\Big)\bigg), \label{eqn:objective_1} \\
            \text{subject to:} \; & 0 \leq a_i \leq 1, \quad \forall i \in \mathcal{J} \label{eqn:constraint_1st}, \\
            & a_0 - \sum_{i \in \mathcal{I}} n_i \cdot h_i(a_i) \leq 0, \label{eqn:constraint_2nd} \\
            & 0 \leq \textcolor{black}{\tau}_0(m) \leq b_0, \label{eqn:constraint_tax_1st} \\
            & \sum_{i \in \mathcal{J}} \textcolor{black}{\tau}_i(m) = 0, \label{eqn:constraint_tax_2nd}
        \end{align}
    where $a = \big(a_0, a_1 \dots, a_I\big)$ and $\textcolor{black}{\tau}(m) = \big(\textcolor{black}{\tau}_0(m), \textcolor{black}{\tau}_1(m),$ $\dots, \textcolor{black}{\tau}_{|\mathcal{I}|}(m)\big)$ denote the action and tax profiles of all players, respectively.
\end{problem}

In Problem \ref{problem1}, \eqref{eqn:constraint_2nd} ensures that the aggregate average trust of all users satisfies the government's lower bound $a_0$,  \eqref{eqn:constraint_tax_1st} 
restricts the government's investment $\textcolor{black}{\tau}_0(m)$ to be within the available budget, and \eqref{eqn:constraint_tax_2nd} ensures that the mechanism is budget balanced.

Note that, in Problem \ref{problem1}, the social planner does not have knowledge about the functional form of either the valuation function $v_i(\cdot)$ of any player $i \in \mathcal{J}$, or the average trust function $h_i(\cdot)$ of any platform $i \in \mathcal{I}$.
If the social planner knew these functions, then she could solve Problem \ref{problem1} using standard optimization methods to allocate the optimal filter $a_i$ and tax $\textcolor{black}{\tau}_i(m)$ to each platform $i \in \mathcal{I}$, and the optimal lower bound $a_0$ and investment $\textcolor{black}{\tau}_0(m)$ to the government. The objective function of Problem \ref{problem1} is differentiable and concave, and the set of feasible solutions is non-empty, convex, and compact. Thus, Problem \ref{problem1} is a convex optimization problem with a unique optimal solution \cite{boyd}. 
However, this solution cannot be computed directly by the social planner because of the private information of the players. Note that if the social planner simply asks the players to report their private information, then the players may not be truthful. Thus, the social planner seeks to design the taxes $\textcolor{black}{\tau}_i(m)$ for each player $i \in \mathcal{J}$ to incentivize the players to be truthful while, at the same time, maximizing the social welfare.


\begin{remark}
    The government has a no compelling reason to misreport to the social planner their budget $b_0$. Thus, we consider that the social planner has knowledge of $b_0$.
\end{remark}

\begin{remark}
    By maximizing the social welfare $u_0(m,a_0) + \sum_{i \in \mathcal{I}}u_i(m,a_k:k\in \mathcal{C}_i)$ in Problem \ref{problem1}, the utility of each player is maximized. Hence, participation of the players in the mechanism is incentivized. Note that the government is not in the position to design the mechanism because they would seek to optimize only their own utility $u_0(m,a_0)$. Thus, the government hires the social planner to design and implement the mechanism described next.
\end{remark}
\color{black}

\section{Mechanism Design Approach} \label{section:md_problem}

\color{black}

In this section, we present a two-step mechanism to incentivize filtering of misinformation among social media platforms. The objective of the first step is to ensure that the social media platforms voluntarily agree to participate in the mechanism. The objectives of the second step are to: (i) extract truthful information from the participating platforms, (ii) derive the optimal level of investment for the government, and (iii) design appropriate taxes for the platforms to maximize the social welfare of the system.

\subsection{Step One - The Participation Step}

\color{black}

In step one of the mechanism, each social media platform $i \in \mathcal{I}$ must decide whether to participate in the mechanism, with complete knowledge of the rules of the second step of the mechanism described in the next subsection. Consider a platform $i \in \mathcal{I}$ that chooses not to participate in the mechanism. Thus, this platform neither pays taxes nor receives any subsidies from the government, i.e., $\textcolor{black}{\textcolor{black}{\tau}_i(m)} = 0$. Furthermore, platform $i$ is free to select the lowest value of $a_i = 0$ that maximizes the valuation $v_i\big(a_k:k\in\mathcal{C}_i\big)$. Meanwhile, another competing platform $l \in \mathcal{C}_{- i}$ may decide to participate in the mechanism and subsequently implement a non-zero filter $a_l$. \textcolor{black}{From Assumption \ref{assumption:excludability}, the government ensures that platform $l$ receives no utility as a result of filter $a_l$. Thus, the utility of the non-participating platform $i$ is given by $v_i\big(a_k = 0: k \in \mathcal{C}_i\big)$. We will use this utility for a non-participating platform in Theorem \ref{thm:ir} of Section IV to establish that all platforms decide to voluntarily participate in step one of the mechanism.}


\color{black}
\subsection{Step Two - The Bargaining Step}

\color{black}
\textcolor{black}{In step two, the social planner asks each player $i \in \mathcal{J}$ to broadcast a message $m_i$ from a set of feasible messages $\mathcal{M}_i$.} For each platform $i \in \mathcal{I}$, let  $\mathcal{D}_{i} = \mathcal{C}_i \cup \{0\}$, and $\mathcal{D}_{-i} = \mathcal{D}_i \setminus \{i\}$.
The message of platform $i$ is defined as
\begin{equation}\label{eqn:defn_message}
    m_i := (\textcolor{black}{\Tilde{h}_i}, \textcolor{black}{\Tilde{p}_i}, \Tilde{a}_i),
\end{equation}
where
$\textcolor{black}{\Tilde{h}_i} \in \mathbb{R}_{\geq 0}$ is the minimum average trust that platform $i$ proposes to achieve through filtering; $\textcolor{black}{\Tilde{p}_i} \in \mathbb{R}_{\geq 0} ^ {|\mathcal{D}_{-i}|}$ is the collection of prices that platform $i$ is willing to pay or receive per unit changes in the filters of other competing platforms \textcolor{black}{(except $i$)} and \textcolor{black}{the government's lower bound}, given by
\begin{equation}\label{eqn:prices}
    \textcolor{black}{\Tilde{p}_i} := (\textcolor{black}{\Tilde{p}_l} ^ i : l \in \mathcal{D}_{-i});
\end{equation}
and $\Tilde{a}_i = (\Tilde{a}_k ^ i: k \in \mathcal{D}_i),$ $\Tilde{a}_i \in \mathbb{R} ^ {|\mathcal{D}_{i}|},$ is the profile of filters for all competing platforms \textcolor{black}{(including $i$)} and \textcolor{black}{government's lower bound proposed by platform $i$.}

\color{black}

\begin{remark}
    Note that each platform proposes a filter for themselves, denoted by $\Tilde{a}_i^i$, in their message $m_i$. However, as it can be seen in \eqref{eqn:prices}, platform $i$ does not propose a price corresponding to $\Tilde{a}_i^i$. This is because we want to give every platform the ability to influence their filter, but not the ability to influence the price associated with their own filter. 
\end{remark}

\color{black}

The message of the government is  $m_0 := (\textcolor{black}{\Tilde{p}_0}, \Tilde{a}_0 ^ 0)$, where $\textcolor{black}{\Tilde{p}_0} \in \mathbb{R}_{\geq 0}$ is the price that the government is willing to pay or receive per unit change of the average trust, and $\Tilde{a}_0^0 \in \mathbb{R}$ is the \textcolor{black}{lower bound} proposed by the government. Note that our mechanism respects the privacy of each platform $i \in \mathcal{I}$ since she does not request either their valuation function $v_i\big(a_k : k \in \mathcal{C}_i \big)$ or their average trust function $h_i(a_i)$. Similarly, the government is not forced to publicly reveal the functional form of their valuation function $v_0(a_0)$. \textcolor{black}{Also each platform $i$ is free to select any feasible values for the components of the message $m_i$.}

\textcolor{black}{Based on the message profile $m := (m_0, m_1,$ $\dots, m_{|\mathcal{I}|})$ that the social planner receives, she allocates the following parameters to the players:}

\textcolor{black}{\textit{1)} The social planner allocates a filter to each platform $i \in \mathcal{I}$ and a lower bound to the government such that the constraints of Problem $1$ are satisfied.}
The filter allocated by the social planner to platform $i$ is $\textcolor{black}{\alpha}_i(m) := \sum_{k \in \mathcal{C}_i} \frac{\Tilde{a}_i ^ k}{|\mathcal{C}_i|}$, i.e., the average of the filters proposed by all competing platforms including $i$. The lower bound allocated by the social planner to the government is $\textcolor{black}{\alpha}_0(m) = \sum_{k \in \mathcal{J}} \frac{\Tilde{a}_0 ^ k}{|\mathcal{J}|}$, i.e., the average of the lower bounds proposed by all platforms and the government. 

\textcolor{black}{\textit{2)}} The social planner allocates a minimum average trust {$\textcolor{black}{\eta_i}(m) \in [0,1]$} to each platform $i \in \mathcal{I}$, given by
\begin{equation}\label{eqn:defn_of_eta}
    \textcolor{black}{\eta_i}(m) = \min\left\{ \frac{n_i \cdot \textcolor{black}{\Tilde{h}_i}}{\sum_{k \in \mathcal{I}} n_k \cdot \textcolor{black}{\Tilde{h}_k}} \cdot \textcolor{black}{\alpha}_0(m), \; 1\right\},
\end{equation}
where the social planner will not accept a message $m_i$ from a platform $i$ that might lead to a situation where $\sum_{k \in \mathcal{I}} n_k \cdot \textcolor{black}{\Tilde{h}_k} = 0$.
\textcolor{black}{The allocated minimum average trust, $\eta_i(m)$, is a lower bound on average trust that must be achieved by platform $i$.}
Let the filter implemented by platform $i$ be $a_i$. Then, platform $i$ must ensure that $n_i \cdot h_i(a_i)\geq \textcolor{black}{\eta_i}(m)$.
Recall \textcolor{black}{from the information structure} that a potential violation of this condition cannot be detected by the social planner since she does not have explicit knowledge of the function $h_i(\cdot)$. However, by Assumption \ref{assumption:assumed_knowledge_b}, the output of $h_i(a_i)$ can be monitored by any other competing platform $l \in \mathcal{C}_{- i}$. 
Any violation of $n_i \cdot h_i(a_i)\geq \textcolor{black}{\eta_i}(m)$ will be reported by platform $l$ to the social planner, in order to ensure that platform $i$ implements the largest filter $a_i,$ and maximizes the utility $u_l(m, a_k:k\in \mathcal{C}_l)$. \textcolor{black}{This prevents platforms from violating the constraint imposed by the allocated minimum average trust $\eta_i(m)$.}

\textcolor{black}{\textit{3)}} The social planner allocates a price $\textcolor{black}{\pi_l}  ^ i := \sum_{k \in \mathcal{C}_{- l} : k \neq i} \frac{\textcolor{black}{\Tilde{p}_l} ^ k}{|\mathcal{C}_{l}| - 2}$, $\textcolor{black}{\pi^i_l} \in \mathbb{R}_{\geq 0},$ to each platform $i \in \mathcal{I}$, corresponding to the allocated filter $\textcolor{black}{\alpha}_l(m)$ of every other competing platform $l \in \mathcal{C}_{- i}$. This price is derived as the average of prices proposed for the allocated filter $\textcolor{black}{\alpha}_l(m)$ by all competing platforms in $\mathcal{C}_{-l}$ except $i$. Thus, the allocated price $\textcolor{black}{\pi_l}^i$ is independent of the prices proposed by both platforms $i$ and $l$. \textcolor{black}{Similarly, the social planner allocates the price $\pi_0 = \sum_{i \in \mathcal{I}} \frac{\textcolor{black}{\Tilde{p}^i_0}}{|\mathcal{I}|}$ to the government.
Note that even though the prices allocated to each player
depend on the message profile $m$, we do not present them with the argument of $m$ to simplify our notation and improve the readability of the subsequent equations.}

\textcolor{black}{\textit{4)}} The social planner \textcolor{black}{allocates} the following tax to each social media platform $i \in \mathcal{I}$,
\begin{gather}
    \textcolor{black}{\tau}_i(m) := - \textcolor{black}{\Tilde{p}_0} \cdot \textcolor{black}{\eta_i}(m) - \sum_{l \in \mathcal{C}_{- i}} \textcolor{black}{\pi_i} ^ l \cdot \textcolor{black}{\alpha}_{i}(m) 
    + \sum_{l \in \mathcal{C}_{- i}} \textcolor{black}{\pi_l} ^ i \cdot \textcolor{black}{\alpha}_{l}(m) \nonumber \\
    + \sum_{l \in \mathcal{C}_{- i} \cup \{0\}} \textcolor{black}{\Tilde{p}_l} ^ i \cdot (\Tilde{a}_l ^ i - \Tilde{a}_l ^ {- i}) ^ 2, \label{eqn:payment_function}
\end{gather}
where $\Tilde{a}_l ^ {- i} = \sum_{k \in \mathcal{C}_{l}:k \neq i} \frac{\Tilde{a}_l ^ k}{|\mathcal{C}_{l}| - 1}$, for each $l \in \mathcal{C}_{- i}$,
is the average of the proposed filters for $l$ by all competing platforms except $i\in\mathcal{I}$, \textcolor{black}{and $\Tilde{a}_0 ^ {- i} = \sum_{k \in \mathcal{J}_{-i}} \frac{\Tilde{a}_0 ^ k}{|\mathcal{J}| - 1}$ is the average of lower bounds proposed by all players except $i$.} The tax $\textcolor{black}{\tau}_i(m)$ of platform $i$ in \eqref{eqn:payment_function} can be \textcolor{black}{interpreted as follows: 
(i) the first term in \eqref{eqn:payment_function} represents a subsidy given by the government to platform $i$ for the increase in average trust among the users of platform $i$; 
(ii) the second term in \eqref{eqn:payment_function} is a collection of subsidies given by each competing platform $l \in \mathcal{C}_{- i}$ to platform $i$ for the increase in valuation $v_l\big(a_k: k \in \mathcal{C}_l\big)$ due to the allocated filter $\textcolor{black}{\alpha}_i$; (iii) the third term in \eqref{eqn:payment_function} is a payment by platform $i$ for the increase in valuation $v_i\big(a_k:k \in \mathcal{C}_i\big)$ due to the allocated filter $\alpha_l$ of each competing platform $l \in \mathcal{C}_{- i}$; and 
(iv) the fourth term in \eqref{eqn:payment_function} is a collection of penalties to platform $i$ \textcolor{black}{if either the filter proposed in message $m_i$ for any competing platform $l \in \mathcal{C}_{- i}$ is inconsistent} with the filters proposed by other platforms, \textcolor{black}{or if the lower bound proposed in $m_i$ is inconsistent with the lower bound proposed by other players.}} Note that the fourth term also penalizes platform $i$ for higher values of proposed prices $\textcolor{black}{\Tilde{p}_l} ^ i$ 
and thus, ensures that the platform $i$ proposes lower prices for the actions of other players. 

Finally, the social planner \textcolor{black}{allocates the following investment to the government:}
\begin{equation}\label{eqn:payment_function_gov}
    \textcolor{black}{\tau}_0(m) = \pi_0 \cdot \textcolor{black}{\alpha}_0(m) + (\textcolor{black}{\Tilde{p}_0} - \pi_0) ^ 2,
\end{equation}
\textcolor{black}{where} the first term is the total investment made by the government for the allocated low bound $\textcolor{black}{\alpha}_0(m)$, and the second term is a penalty when the price proposed by the government deviates from the price allocated to the government.

\color{black}
\begin{remark}
Note that in \eqref{eqn:payment_function}, for some filter $a_i>0$ of platform $i$, the social planner takes a payment from each competing platform $l \in \mathcal{C}_{-i}$ and allocates an equal subsidy to platform $i$. This subsidy serves a dual purpose: (i) it incentivizes platform $i$ to implement the filter $a_i$, and (ii) it eventually leads to a fair distribution of the government's investment among all platforms.
\end{remark}

\begin{remark}
    We presented the step two of the mechanism under the implicit assumption that all social media platforms participate in the mechanism. This does not cause any implications, however, since, as we prove in Theorem \ref{thm:ir} next, all platforms eventually, indeed, participate in the mechanism in step one.
\end{remark}

\color{black}

\textcolor{black}{The step two of the mechanism is characterized by the tuple} $\langle \mathcal{M}, g(\cdot) \rangle$, where $\mathcal{M} = \mathcal{M}_0 \times \mathcal{M}_1 \times \dots \times \mathcal{M}_{|\mathcal{I}|}$ is the complete message space of all \textcolor{black}{players}, and $g(\cdot) : \mathcal{M} \to \mathcal{O}$ is the outcome function that maps each message profile to a set of outcomes $\mathcal{O}$. The set of outcomes is in the form
\begin{multline}
    \mathcal{O} := \Big\{\big(\textcolor{black}{\alpha}_0(m), \textcolor{black}{\alpha}_1(m), \dots, \textcolor{black}{\alpha}_{|\mathcal{I}|}(m)\big), \big(\textcolor{black}{\tau}_0(m), \textcolor{black}{\tau}_1(m), \\
    \dots, \textcolor{black}{\tau}_{|\mathcal{I}|}(m)\big) : \textcolor{black}{\alpha}_i(m) \in \mathcal{A}, \; \textcolor{black}{\tau}_i(m) \in \mathbb{R}, \; i\in\mathcal{J}\Big\},
\end{multline}
and the outcome function $g(m)$ determines the outcome of any given message profile $m = (m_0, m_1, \dots, m_I) \in \mathcal{M}$.

\subsection{Generalized Nash Equilibrium and the Induced Game}

Formally, a mechanism $\langle \mathcal{M}, g(\cdot) \rangle$ together with the utility functions $(u_i)_{i \in \mathcal{I}}$ induces a game in which \textcolor{black}{the social planner allocates the filters $(\textcolor{black}{\alpha}_1(m), \dots, \textcolor{black}{\alpha}_i(m))$ to the platforms and the lower bound $\textcolor{black}{\alpha}_0(m)$ to the government.} Each platform $i \in \mathcal{I}$ that participates in the mechanism must implement the filter $a_i = \textcolor{black}{\alpha}_i(m)$, and the government must select the lower bound $a_0 = \textcolor{black}{\alpha}_0(m)$.
\textcolor{black}{Note that platform $i$ can influence their allocated filter $\alpha_i(m)$ with their message $m_i$.} Thus, the strategy of platform $i$ in the induced game is given by the message $m_i \in \mathcal{M}_i$ \cite{mas_colell1995}, with a constraint that $\textcolor{black}{\alpha}_i(m) \in \mathcal{S}_i(m)$, where
\begin{equation}
    \mathcal{S}_i(m) = \{a_i \in \mathcal{A} : n_i \cdot h_i(a_i) \geq \textcolor{black}{\eta_i}(m)\}.
\end{equation}
Thus, the set of feasible allocations $\mathcal{S}_i(m)$ for $i\in \mathcal{I}$ is a function of the messages of all social media in $\mathcal{I}$ and the government. The strategy of the government is denoted by the message $m_0$ and the set of feasible strategies is given by $\mathcal{M}_0$. For such a game, we select the solution concept of the generalized Nash equilibrium (GNE) \cite{facchinei2010generalized}. Let $m_{- i} = (m_0, \dots, m_{i - 1}, m_{i + 1}, \dots, m_{I})$.
A message profile $m ^ * = (m_i ^ * : i \in \mathcal{J})$ is the GNE of the induced game, if (i) for each $i \in \mathcal{I}$, 
\color{black}
\begin{multline} \label{eqn:GNE_defn}
    u_i\big((m_i ^ *, m_{- i} ^ *), \textcolor{black}{\alpha}_k(m_i ^ *, m_{- i} ^ *) : k \in \mathcal{C}_i\big) \\
    \geq u_i\big((m_i, m_{- i} ^ *), \textcolor{black}{\alpha}_k(m_i, m_{- i} ^ *) : k \in \mathcal{C}_i\big),
\end{multline}
\color{black}
for all $m_i \in \mathcal{M}_i$ and $\textcolor{black}{\alpha}_i \in \mathcal{S}_i(m)$; and (ii) the message $m_0 ^ *$ of the government is such that $u_0\big((m_0 ^ *, m_{- 0} ^ *), \textcolor{black}{\alpha}_0(m_0 ^ *, m_{- 0} ^ *)\big) \geq u_0\big((m_0, m_{- 0} ^ *), \textcolor{black}{\alpha}_0(m_0, m_{- 0} ^ *)\big),$ for all $m_0 \in \mathcal{M}_0$. 
\color{black}
To simplify the notation, in the remaining of the paper, we denote the utility of platform $i \in \mathcal{I}$ by $u_i(m_i,m_{-i})$ and the utility of the government by $u_0(m_0, m_{-0})$.

\begin{remark}
    In general, the GNE solution concept is defined for a game with complete information. However, we adopt this solution in our induced game despite the fact that the valuation function $v_i\big(a_k : k \in \mathcal{C}_i\big)$ and the average trust function $h_i(a_i)$ are the private information of platform $i$. We resolve this discrepancy by considering that the induced game is played repeatedly over multiple iterations, and thus, the social media platforms can utilize an iterative learning process to find a GNE. This interpretation of a GNE is consistent with the theory of mechanism design \cite{groves_b}.
\end{remark}


\subsection{Summary of the Notation}

We summarize the variables introduced in Sections II and III in Table \ref{summary_table}. As a general guideline, we use lowercase letters of the English alphabet to denote variables and functions, lowercase letters with tilde to denote variables in a message, and lowercase letters of the Greek alphabet to indicate variables allocated to the players by the social planner. We use scripted letters to denote sets.
Furthermore, we use $\mathbb{R}$ to denote the set of real numbers, $\mathbb{R}_{\geq0}$ to denote the set of non-negative real numbers, and $\mathbb{N}$ to denote the set of natural numbers.

\begin{table}[t]
\color{black}
\caption{A summary of the key variables}
\label{summary_table}
\begin{tabular}{ p{30pt} p{195pt} }
 \hline
 \textbf{Symbol} &  \multicolumn{1}{c}{\textbf{Explanation}}\\[3pt]
 \hline
 $m_i$ &The message broadcast by player $i \in \mathcal{J}$ \\ [3pt]
 $a_i$ & The filter of platform $i \in \mathcal{I}$\\[3pt]
 $\Tilde{a}^i_k$ & The filter proposed by platform $i \in \mathcal{I}$ for platform $k \in \mathcal{C}_i$ \\[3pt]
 $\alpha_i(m)$ & The filter allocated to platform $i \in \mathcal{I}$\\[3pt]
 $a_0$ & The government's lower bound on trust \\[3pt]
 $\Tilde{a}_0$ & The lower bound proposed by the government \\[3pt]
 \multirow{2}{*}{$\Tilde{a}^i_0$} & The lower bound proposed by platform $i \in \mathcal{I}$ for the government \\[3pt]
 $\alpha_0(m)$ & The lower bound allocated to the government \\ [3pt]
  %
 $v_i(\cdot)$ & The valuation function of player $i \in \mathcal{J}$ \\ [3pt]
 %
 $h_i(\cdot)$ & The average trust function of platform $i \in \mathcal{I}$ \\ [3pt]
  %
 $\Tilde{h}_i$ & The proposed minimum average trust of platform $i \in \mathcal{I}$ \\ [3pt]
  %
 $\eta_i(m)$ & The allocated minimum average trust for platform $i \in \mathcal{I}$ \\ [3pt]
  %
 \multirow{2}{*}{$\Tilde{p}^i_l$} & The price proposed by platform $i \in \mathcal{I}$ corresponding to  player $l \in \mathcal{D}_{-i}$ \\ [3pt]
 \multirow{2}{*}{$\pi^i_l$} & The price allocated to platform $i \in \mathcal{I}$ corresponding to player $l \in \mathcal{D}_{-i}$ \\ [3pt]
 $\Tilde{p}_0$ & The price proposed by the government \\ [3pt]
  %
 $\pi_0$ & The price allocated to the government \\ [3pt]
  %
 %
 $\textcolor{black}{\tau}_i(m)$ &The tax allocated to player $i \in \mathcal{J}$\\
 \hline
\end{tabular}
\color{black}
\end{table}

\color{black}

\section{Properties of the Mechanism} \label{section:properties_of_mechanism}

In this section, we show that our proposed mechanism has the following desirable properties: (i) budget balance at GNE, (ii) feasibility at GNE, (iii) strong implementation, (iv) existence of at least one GNE, and (v) individual rationality. 

\textcolor{black}{Recall that each social media platform $i \in \mathcal{I}$ is a strategic player who seeks to maximize their utility $u_i(m_i, m_{-i})$ through the choice of message $m_i \in \mathcal{M}_i$. Thus, we can define the following optimization problem from the perspective of platform $i \in \mathcal{I}$ in the induced game.}

\begin{problem}\label{problem2}
    The optimization problem for social media platform $i \in \mathcal{I}$ in the induced game is
        \begin{align}
            \max_{m_i \in \mathcal{M}_i} \; & v_i\big(\textcolor{black}{\alpha}_k(m) : k \in \mathcal{C}_{-i}\big) - \textcolor{black}{\tau}_i(m), \label{eqn:objective_nash} \\
            \text{subject to: } & 0 \leq \textcolor{black}{\alpha}_i(m) \leq 1, \label{eqn:constraint_nash_1st} \\
            & \textcolor{black}{\eta_i}(m) - n_i \cdot h_i\big(\textcolor{black}{\alpha}_i(m)\big) \leq 0, \label{eqn:constraint_nash_2nd}
        \end{align}
    \textcolor{black}{where the objective function in \eqref{eqn:objective_nash} is the utility $u_i(m_i, m_{-i})$ of platform $i$,} 
    \eqref{eqn:constraint_nash_1st} ensures that the allocated filter of platform $i$ is feasible, 
    and \eqref{eqn:constraint_nash_2nd} ensures that the fraction of average trust among users of platform $i$ is greater than the minimum average trust allocated by the social planner.
\end{problem}
Note that the social planner can ensure that \eqref{eqn:constraint_nash_1st} and \eqref{eqn:constraint_nash_2nd} are hard constraints by imposing a tax $\textcolor{black}{\tau}_i(m) \rightarrow \infty$ when they are violated. \textcolor{black}{Next, recall that the government is also a strategic player in the induced game who seeks to maximize their utility $u_0(m_0,m_{-0})$ through the choice of message $m_0 \in \mathcal{M}_0$.}

\begin{problem}\label{problem3}
    The optimization problem for the government is
        \begin{align}
            \max_{m_0 \in \mathcal{M}_0} \; & v_0\big(\textcolor{black}{\alpha}_0(m)\big) - \textcolor{black}{\tau}_0(m), \label{eqn:objective_nash_gov} \\
            \text{subject to: } & 0 \leq \textcolor{black}{\alpha}_0(m) \leq 1, \label{eqn:constraint_nash_1st_gov} \\
            & \pi_0 \cdot \textcolor{black}{\alpha}_0(m) - b_0 \leq 0, \label{eqn:constraint_nash_2nd_gov}
        \end{align}
        where \textcolor{black}{the objective in \eqref{eqn:objective_nash_gov} is the utility $u_0(m_0, m_{-0})$ of the government,} \eqref{eqn:constraint_nash_1st_gov} ensures that the government's lower bound $a_0$ is feasible, 
        and \eqref{eqn:constraint_nash_2nd_gov} ensures that the total government's investment is less than their budget $b_0$.
\end{problem}

\color{black}
\begin{remark}
    Consider an optimal solution $m_i^* \in \mathcal{M}_i$ of Problem \ref{problem2} for each platform $i \in \mathcal{I}$, and an optimal solution $m_0^* \in \mathcal{M}_0$ of Problem \ref{problem3} for the government. The message profile $m^* = \big(m_0^*, m_1^*, \dots, m_{|\mathcal{I}|}^*\big) \in \mathcal{M}$ satisfies \eqref{eqn:GNE_defn}, and thus, forms a GNE of the induced game. 
\end{remark}

\color{black}

\textcolor{black}{Next, we establish some basic properties of the mechanism in Lemmas \ref{lemma:truthful_prices} and \ref{lemma:truthful_action_proposals} at any GNE, if one exists. In Lemma \ref{lemma:truthful_prices}, we refer to Problem \ref{problem3} to show} that the government's proposed price at any GNE of the induced game is equal to the average price proposed by all social media.

\begin{lemma}\label{lemma:truthful_prices}
    Let the message profile $m ^ * \in \mathcal{M}$ be a GNE of the induced game. Then, $\textcolor{black}{\Tilde{p}_0} ^ * = \pi_0 ^ *$ for the government.
\end{lemma}

\begin{proof}
    Since the objective function in Problem \ref{problem3} is concave with respect to the price $\textcolor{black}{\Tilde{p}_0}$, the price $\textcolor{black}{\Tilde{p}_0} ^ *$ at GNE can be using the equation $\frac{\partial u_0}{\partial \textcolor{black}{\Tilde{p}_0}} \big|_{\textcolor{black}{\Tilde{p}_0} ^ *} =$ $2 \cdot (\textcolor{black}{\Tilde{p}_0} ^ * - \pi_0 ^ *) = 0,$
    which yields $\textcolor{black}{\Tilde{p}_0} ^ * = {\pi}_{- 0} ^ *$.
\end{proof}

\textcolor{black}{Similarly, in the next result (Lemma \ref{lemma:truthful_action_proposals}), we refer to Problem \ref{problem2} to establish that, at any GNE, the filters proposed by all social media platforms in $\mathcal{C}_i$ for platform $i$} are the equal, unless the corresponding price proposal is $0$. \textcolor{black}{Furthermore, at every GNE, if one exists, the lower bound proposed by all platforms is the same, unless the corresponding price proposal is $0$.}

\begin{lemma}\label{lemma:truthful_action_proposals}
    Let the message profile $m ^ * \in \mathcal{M}$ be a GNE of the induced game. Then, for $\textcolor{black}{\Tilde{p}_k} ^ i \neq 0$, we have $\Tilde{a}_k ^ {i *} = \Tilde{a}_k ^ {- i *}$ for every social media platform $i \in \mathcal{I}$, for every $k \in \mathcal{D}_{-i}$.
\end{lemma}

\begin{proof}
    The proof is similar to the proof of Lemma \ref{lemma:truthful_prices}, and thus, \textcolor{black}{due to space limitations, it is omitted.}
\end{proof}

\textcolor{black}{Next, we use the properties established in Lemmas \ref{lemma:truthful_prices} and \ref{lemma:truthful_action_proposals} to show that our proposed mechanism is budget balanced at any GNE, if one exists, i.e., the social planner redistributes all the payments it collects from the players as subsidies to the players.}

\begin{theorem}[\textbf{Budget Balance}] \label{budget_balance}
    Consider any GNE $m ^ * \in \mathcal{M}$ of the induced game. Then, the proposed mechanism is budget balanced, i.e., $\sum_{i \in \mathcal{J}} \textcolor{black}{\tau}_i(m ^ *) = 0$.
\end{theorem}

\begin{proof}
    From Lemmas \ref{lemma:truthful_prices} and \ref{lemma:truthful_action_proposals}, the tax $\textcolor{black}{\tau}_i ^ * = \textcolor{black}{\tau}_i(m ^ *)$ for social media platform $i$ at GNE is $\textcolor{black}{\tau}_i ^ * = - \textcolor{black}{\Tilde{p}_0} ^ * \cdot \textcolor{black}{\eta_i}(m ^ *) - \sum_{l \in \mathcal{C}_{- i}} \textcolor{black}{\pi_i} ^ l \cdot \textcolor{black}{\alpha}_{i}(m^*)  + \sum_{l \in \mathcal{C}_{- i}} \textcolor{black}{\pi_l} ^ {i} \cdot \textcolor{black}{\alpha}_l(m^*).$ The tax $\textcolor{black}{\tau}_0 ^ *$ for the government at GNE is $\textcolor{black}{\tau}_0 ^ * = \textcolor{black}{\Tilde{p}_0} ^ * \cdot \textcolor{black}{\alpha}_0(m ^ *),$ where $\textcolor{black}{\Tilde{p}_0}^*$ is the price per unit change on average trust at GNE.
    Since $\sum_{i \in \mathcal{I}} \textcolor{black}{\eta_i}(m) = \textcolor{black}{\alpha}_0(m),$ for all $m \in \mathcal{M}$, then at GNE we have
        $
            \sum_{i \in \mathcal{J}} \textcolor{black}{\tau}_i ^ * = \sum_{i \in \mathcal{I}} \Big[- \sum_{l \in \mathcal{C}_{- i}} \textcolor{black}{\pi_i} ^ l \cdot \textcolor{black}{\alpha}_i(m^*)
            + \sum_{l \in \mathcal{C}_{- i}} \textcolor{black}{\pi_l} ^ {i} \cdot \textcolor{black}{\alpha}_l(m^*)\Big]=0.
            $
\end{proof}
\color{black}

In the next result (Lemma \ref{lemma:feasibility}), we establish that every GNE, $m^* \in \mathcal{M}$, if one exists, of the induced game leads to an allocation of filters for the platforms and a lower bound for the government that forms a feasible solution of Problem \ref{problem1}. In other words, every GNE of the induced game ensures that all constraints of Problem \ref{problem1} are satisfied.

\color{black}

\begin{lemma}[\textbf{Feasibility}] \label{lemma:feasibility}
    Every GNE message profile $m ^ * \in \mathcal{M}$ leads to a filter profile $\big(\textcolor{black}{\alpha}_1(m ^ *), \dots, \textcolor{black}{\alpha}_{|\mathcal{I}|}(m ^ *)\big)$ and lower bound $\textcolor{black}{\alpha}_0(m ^ *)$, which is a feasible solution of Problem \ref{problem1}.
\end{lemma}

\begin{proof}
    Every GNE message profile $m ^ *$ satisfies \eqref{eqn:constraint_nash_1st} - \eqref{eqn:constraint_nash_2nd} and \eqref{eqn:constraint_nash_1st_gov} - \eqref{eqn:constraint_nash_2nd_gov}. From Theorem \ref{budget_balance}, $\sum_{i \in \mathcal{J}} \textcolor{black}{\tau}_i(m ^ *) = 0$. For each $i \in \mathcal{I}$, $\textcolor{black}{\eta_i}(m) \leq n_i \cdot h_i(\textcolor{black}{\alpha}_i(m)),$ and $\sum_{i \in \mathcal{I}} \textcolor{black}{\eta_i}(m) = \textcolor{black}{\alpha}_0(m)$. Hence, $\sum_{i \in \mathcal{I}} h_i(\textcolor{black}{\alpha}_i(m)) \geq \textcolor{black}{\alpha}_0(m).$
\end{proof}

In the next result (Lemma \ref{lemma:achievability}), we establish that every social media platform $i\in \mathcal{I}$ can unilaterally deviate in the message $m_i \in \mathcal{M}_i$, to achieve any desired allocation of filters for every competing platform, including itself. \textcolor{black}{This property of our mechanism ensures that each platform $i \in \mathcal{I}$ can attain any filter $\hat{a}_i \in \mathcal{A}$, irrespective of the filters proposed by the competing platforms.}

\begin{lemma} \label{lemma:achievability}
    Given the message profile $m_{- i} \in \mathcal{M}_{- i}$, the social media platform $i \in \mathcal{I}$ can unilaterally deviate in their message $m_i \in \mathcal{M}_i$ to attain any filter \textcolor{black}{$\hat{a}_k \in \mathcal{A}$ as the allocated filter $\textcolor{black}{\alpha}_k(m) \in \mathcal{S}_k(m)$, for all $k \in \mathcal{C}_{i}$.}
\end{lemma}

\begin{proof}
    Let $m_{- i} = \big(m_0, \dots, m_{i-1}, m_{i+1}, \dots, m_{|\mathcal{I}|}\big)$ be the message profile of all players in $\mathcal{J}_{-i}$. Then, platform $i$ can propose a filter $\Tilde{a}^i_k = \hat{a}_k - \sum_{l \in \mathcal{C}_{k} : l \neq i} \frac{\Tilde{a}^l_k}{|\mathcal{C}_{k}| - 1}$,
    \color{black}
    to ensure that $\textcolor{black}{\alpha}_k(m) = \hat{a}_k$ for each $k \in \mathcal{C}_i$. Moreover, platform $i$ can propose a lower bound $\Tilde{a}^i_0 = - \sum_{l \in \mathcal{J}_{-i}}\Tilde{a}^l_0$ for the government, 
    to ensure that $\textcolor{black}{\alpha}_0(m) = 0$, and subsequently, $\textcolor{black}{\alpha}_k(m) = \hat{a}_k \in \mathcal{S}_k(m)$ for all $k \in \mathcal{C}_i$.
\end{proof}

Next, we establish that, at any GNE, \textcolor{black}{if one exists,} of the induced game the allocated filters for all platforms and the allocated lower bound for the government result in the optimal solution of Problem \ref{problem1}.

\begin{theorem}[\textbf{Strong Implementation}] \label{thm:implementation}
    Consider any GNE $m ^ * \in \mathcal{M}$ of the induced game. Then, the allocated filter profile $\big(\textcolor{black}{\alpha}_1(m ^ *), \dots, \textcolor{black}{\alpha}_{|\mathcal{I}|}(m ^ *)\big)$ and the allocated lower bound $\textcolor{black}{\alpha}_0(m ^ *)$ at equilibrium is equal to the optimal solution $a ^ {*o}$ of Problem \ref{problem1}.
\end{theorem}

\begin{proof}
    Let $\textcolor{black}{\alpha}(m ^ *) = \big(\textcolor{black}{\alpha}_1(m ^ *), \dots, \textcolor{black}{\alpha}_{|\mathcal{I}|}(m ^ *)\big)$. Then, the GNE message profile $m ^ *$ satisfies, for platform $i \in \mathcal{I}$, the following Kush-Kahn-Tucker (KKT) conditions for optimality:
        \begin{align}
            \frac{\partial v_i}{\partial \textcolor{black}{\alpha}_i} \Bigg|_{\textcolor{black}{\alpha}(m ^ *)} + \sum_{l \in \mathcal{I}_{- i}} \textcolor{black}{\pi^i_l} - \lambda_i ^ i + \mu_i ^ i + \nu_i ^ {i} & \cdot \frac{\partial h_i}{\partial \textcolor{black}{\alpha}_i} \Bigg|_{\textcolor{black}{\alpha}(m ^ *)} = 0, \label{eqn:implementation_1} \\
            \frac{\partial v_i}{\partial \textcolor{black}{\alpha}_l} \Bigg|_{\textcolor{black}{\alpha}(m ^ *)} - \textcolor{black}{\pi_l}^{i} &= 0, \quad \forall l \in  \mathcal{C}_{- i}, \label{eqn:implementation_1_2} \\
        %
        %
            \textcolor{black}{\Tilde{p}_0} ^ * - \nu_i ^ i & = 0, \label{eqn:implementation_2} \\
            \lambda_i ^ i \cdot \big(\textcolor{black}{\alpha}_i(m ^ *) - 1\big) & = 0, \label{eqn:implementation_3} \\
            \mu_i ^ i \cdot \textcolor{black}{\alpha}_i(m ^ *) & = 0, \label{eqn:implementation_4} \\
            \nu_i ^ i \cdot \big(\textcolor{black}{\eta_i}(m ^ *) - h_i(\textcolor{black}{\alpha}_i(m ^ *))\big) & = 0, \label{eqn:implementation_5} \\
            \lambda_i ^ i, \mu_i ^ i, \nu_i ^ i & \geq 0, \label{eqn:implementation_6}
        \end{align}
    where \eqref{eqn:implementation_1} - \eqref{eqn:implementation_2} are the derivatives of the Lagrangian of platform $i$ with respect to $\textcolor{black}{\alpha}(m)$ and $\textcolor{black}{\eta_i}(m)$, for Problem \ref{problem2}, and \eqref{eqn:implementation_3} - \eqref{eqn:implementation_6} are constraints on the Lagrange multipliers $(\lambda_i^i,\mu_i^i,\nu_i^i)$. From \eqref{eqn:implementation_2},  $\nu_i ^ i = \textcolor{black}{\Tilde{p}_0} ^ *$ for all $i \in \mathcal{I}$.
    Substituting \eqref{eqn:implementation_1_2} in \eqref{eqn:implementation_1}, we have
        \begin{equation}
            \sum_{k \in \mathcal{C}_{i}} \frac{\partial v_k}{\partial \textcolor{black}{\alpha}_i} \Bigg|_{\textcolor{black}{\alpha}(m ^ *)} - \lambda_i ^ i + \mu_i ^ i + \nu_i ^ i \cdot \frac{\partial h_i}{\partial \textcolor{black}{\alpha}_i} \Bigg|_{\textcolor{black}{\alpha}(m ^ *)} = 0,
        \end{equation}
        for all $i \in \mathcal{I}$. Similarly, the KKT conditions for Problem \ref{problem3} are:
        \begin{align}
            \frac{\partial v_0}{\partial \textcolor{black}{\alpha}_0} \Bigg|_{\textcolor{black}{\alpha}_0(m ^ *)} - \textcolor{black}{\Tilde{p}_0} ^ * - \lambda_0 ^ 0 + \mu_0 ^ 0 + \omega_0 ^ 0 \cdot \textcolor{black}{\Tilde{p}_0} ^ * & = 0, \label{imp_gov_1}\\
            \lambda_0 ^ 0 \cdot \big(\textcolor{black}{\alpha}_0(m ^ *) - 1\big) & = 0, \label{imp_gov_2} \\
            \mu_0 ^ 0 \cdot \textcolor{black}{\alpha}_0(m ^ *) & = 0, \label{imp_gov_3}\\
            \omega_0 ^ 0 \cdot \big(\textcolor{black}{\Tilde{p}_0} ^ * \cdot \textcolor{black}{\alpha}_0(m ^ *) - b_0\big) & = 0, \label{im_gov_4}\\
            \lambda_0 ^ 0, \mu_0 ^ 0, \omega_0 ^ 0 & \geq 0, \label{imp_gov_5}
        \end{align}
    where \eqref{imp_gov_1} is the derivative of the Lagrangian, and  \eqref{imp_gov_2} - \eqref{imp_gov_5} are constraints on the Lagrange multipliers $(\lambda_0^0, \mu_0^0, \omega_0^0)$.
    
    The optimal solution $a ^ {*o} = \big(a ^ {*o}_0, a ^ {*o}_1, \dots, a ^ {*o}_{|\mathcal{I}|} \big)$ of Problem \ref{problem1} satisfies the following KKT conditions:
        \begin{align}
            \sum_{k \in \mathcal{C}_{i}} \frac{\partial v_k}{\partial a_i} \Bigg|_{a_i ^ {*o}} - \lambda_i + \mu_i + & \nu \cdot \frac{\partial h_i}{\partial a_i} \Bigg|_{a_i ^ {*o}} = 0, \quad \forall i \in \mathcal{I}, \label{imp_central_1}\\
            \frac{\partial v_0}{\partial a_0} \Bigg|_{a_0 ^ {*o}} - \lambda_0 + \mu_0 & - \nu - \omega \cdot \pi_0 = 0, \label{imp_central_2}\\
            \lambda_i \cdot (a_i ^ {*o} - 1) & = 0, \qquad \forall i \in \mathcal{J}, \label{imp_central_3}\\
            \mu_i \cdot a_i ^ {*o} & = 0, \qquad \forall i \in \mathcal{J}, \label{imp_central_4}\\
            \nu \cdot \big(a_0 ^ {*o} - h_i(a_i ^ {*o})\big) & = 0, \label{imp_central_5}\\
            \omega \cdot (\pi_0 \cdot a_0 ^ {*o} - b_0) & = 0, \label{imp_central_6}\\
            \lambda_i, \mu_i, \omega, \nu & \geq 0, \qquad \forall i \in \mathcal{J}, \label{imp_central_7}
        \end{align}
    where \eqref{imp_central_1} - \eqref{imp_central_2} are the derivatives of the Lagrangian, and \eqref{imp_central_3} - \eqref{imp_central_4} are constraints on the Lagrange multipliers $(\lambda_i, \mu_i, \omega, \nu : i \in \mathcal{J})$.
    By setting $\pi_0 = \textcolor{black}{\Tilde{p}_0} ^ *$,  $\lambda_i = \lambda_i ^ i,$ $\mu_i = \mu_i ^ i,$ $\nu = \textcolor{black}{\Tilde{p}_0} ^ *,$ $\omega = \omega_0 ^ 0,$
    $a_i ^ {*o} = \textcolor{black}{\alpha}_i (m ^ *)$, which implies that the efficient allocation of filters for all platforms and lower bound for the government is implemented by all GNE of the induced game.
\end{proof}

Next, we show that our mechanism guarantees the existence of at least one GNE \textcolor{black}{for the induced game. This ensures that the results of Lemmas \ref{lemma:truthful_prices} - \ref{lemma:feasibility} and Theorems \ref{budget_balance} - \ref{thm:implementation} are always valid for the induced game.}

\begin{theorem}[\textbf{GNE existence}] \label{thm:gne_existence}
    Let $a ^ {*o} = \big(a ^ {*o}_0, a ^ {*o}_1, \dots, a ^ {*o}_{|\mathcal{I}|} \big)$ be the unique optimal solution of Problem \ref{problem1}. Then, there is a GNE message profile $m^* \in \mathcal{M}$ of the induced game that guarantees that the filter profile $\big( \textcolor{black}{\alpha}_1(m ^ *), \dots, \textcolor{black}{\alpha}_{|\mathcal{I}|}(m ^ *)\big)$ and lower bound $\textcolor{black}{\alpha}_0(m ^ *)$ at GNE satisfy $\textcolor{black}{\alpha}_i(m ^ *) = a_i ^ {*o}$, for all $i \in \mathcal{J}$.
\end{theorem}

\begin{proof}
    Consider that the optimal solution $a ^ {*o}$ which satisfies the KKT conditions for Problem \ref{problem1} with the corresponding Lagrange multipliers $(\lambda_i, \mu_i, \nu, \omega : i \in \mathcal{J}).$ Taking similar steps to the proof of Theorem \ref{thm:implementation}, we can show that for $\textcolor{black}{\Tilde{p}_0} = \pi_0 = \nu$, the Lagrange multipliers of Problems \ref{problem2} and \ref{problem3} are  $\lambda_i ^ i = \lambda_i, \; \mu_i ^ i = \mu_i, \; \nu_i ^ i = \nu, \; \omega_0^0 = \omega, \; i \in \mathcal{J},$  and the allocated prices  are $\textcolor{black}{\pi_l}^{i} = \frac{\partial v_i}{\partial \textcolor{black}{\alpha}_l} \big|_{a^{*o}},$ for all $l \in \mathcal{C}_{- i}.$ This implies that the allocated filters at GNE are $\textcolor{black}{\alpha}_i(m ^ *) = a_i ^ {*o}$ for all platforms $i \in \mathcal{I}$, and \textcolor{black}{the allocated lower bound of the government is $\textcolor{black}{\alpha}_0(m ^ *) = a_0 ^ {*o}$.}
\end{proof}
\color{black}
Next, we consider the step one (the participation step) of our mechanism from Section III-A. We first note that the government always participates in the mechanism for the opportunity to incentivize misinformation filtering among the platforms. In the following result (Theorem \ref{thm:ir}), we invoke Assumption \ref{assumption:excludability} and the properties of our mechanism, to show that in step one, every social media platform voluntarily decides to participate in the mechanism. This property is also called individual rationality of the mechanism as it ensures voluntary participation of rational players without dictatorship. 
\color{black}
\begin{theorem}[\textbf{Individually Rational}] \label{thm:ir}
    The proposed mechanism is individually rational, i.e., each platform $i \in \mathcal{I}$ prefers the outcome of every GNE of the induced game to the outcome of not participating.
\end{theorem}

\begin{proof}
    Consider any GNE message profile $m ^ *$. By Lemma \ref{lemma:achievability}, given profile $m_{- i} ^ *$, there exists a message $m_i \in \mathcal{M}_i$ for platform $i$ such that $\textcolor{black}{\alpha}_0(m_i, m_{- i} ^ *) = 0$. Furthermore, platform $i$ can unilaterally deviate in their message $m_i$ to ensure that for every platform $k \in \mathcal{C}_{i}$, the allocated filter is given by $\textcolor{black}{\alpha}_k(m_i, m_{- i} ^ *) = 0$. \textcolor{black}{Assumption \ref{assumption:excludability} implies that the utility of a non-participating platform $i \in \mathcal{I}$ is given by $v_i(a_k = 0: k \in \mathcal{C}_i)$.}
    Consider the message $m_i = (\textcolor{black}{\Tilde{h}_i}, \textcolor{black}{\Tilde{p}_i}, \Tilde{a}_i)$ defined in \eqref{eqn:defn_message} with $\textcolor{black}{\Tilde{p}_l} ^ i = 0,$ for all $l \in \mathcal{C}_{- i} \cup \{0\}, $
        $\Tilde{a}_k ^ i = - \sum_{l \in \mathcal{C}_{- i}} \Tilde{a}_k ^ l,$ for all $k \in \mathcal{C}_{- i},$ and $\Tilde{a}_0 ^ i = - \sum_{l \in \mathcal{J}_{- i}} \Tilde{a}_0 ^ i.$
    Then, the allocation $\textcolor{black}{\alpha}_k(m_i, m_{- i} ^ *) = 0$ is feasible for every platform $k \in \mathcal{C}_{i}$ and the corresponding tax for social media platform $i$ is given by $\textcolor{black}{\tau}_i = 0$. The utility $u_i(m_i, m_{- i} ^ *)$ of social media platform $i$ is given by $u_i(m_i, m_{- i} ^ *) = v_i(0, \dots, 0) - 0.$ From the definition of the GNE in $\eqref{eqn:GNE_defn}$, we have $u_i(m ^ *) \geq u_i(m_i, m_{- i} ^ *)$. Hence, $u_i(m ^ *) \geq v_i(0, \dots, 0).$
    We observe that the utility $u_i(m ^ *)$ at any GNE $m^* \in \mathcal{M}$ of a platform $i \in \mathcal{I}$, that decides to participate in the mechanism, is equal to or greater than their utility when not participating in the mechanism. Thus, in step one of the mechanism, the weakly dominant action of every social media platform $i \in \mathcal{I}$ is to participate in the mechanism.
\end{proof}

\subsection{Extension to Quasi-Concave Valuations}

In this subsection, we relax Assumptions \ref{assumption:filter_compatibility} - \ref{assumption:government_valuation}, and replace them with the following more general assumptions: (i) The valuation function \textcolor{black}{$v_i\big(a_k:k\in\mathcal{C}_i\big): \mathcal{A}^{|\mathcal{C}_i|} \to \mathbb{R}_{\geq0}$} of every platform $i \in \mathcal{I}$ is quasi-concave, differentiable, and have the same monotonic properties as before. (ii) The valuation function \textcolor{black}{$v_0(a_0): \mathcal{A} \to \mathbb{R}_{\geq0}$} of the government is quasi-concave, differentiable and increasing with respect to $a_0$. (iii) The average trust function \textcolor{black}{$h_i(a_i): \mathcal{A} \to [0,1]$} of any social media platform $i \in \mathcal{I}$ is a differentiable and increasing with respect to $a_i$. We cannot use the KKT conditions to prove the existence of a GNE and strong implementation under these relaxed assumptions. However, note that at any GNE, if one exists, the proposed mechanism is still budget balanced, feasible and individually rational. In addition, Lemmas \ref{lemma:truthful_prices}, \ref{lemma:truthful_action_proposals}, and \ref{lemma:achievability} also hold as they do not depend on the concavity of the valuation.

Next, we prove that for the relaxed assumptions, there exists a GNE and that it induces a Pareto efficient equilibrium in the game. Pareto efficiency refers to the condition where we cannot improve the utility of any player without decreasing the utility of another player in the induced game \cite{mas_colell1995}. Pareto efficiency is a weaker property in comparison to the strong implementation achieved by our mechanism for concave valuation functions.

\begin{theorem} \label{thm:quasi}
    Let the valuation function $v_i(a_k:k\in\mathcal{C}_i)$ be quasi-concave and differentiable for all players $i \in \mathcal{J}$ and consider the game $\langle \mathcal{M}, g(\cdot), (u_i)_{i \in \mathcal{I}} \rangle$. Then, (i) there exists a GNE for the induced game, and (ii) every GNE of the induced game is Pareto efficient.
\end{theorem}

\begin{proof}
    \textit{1) Existence:} Consider the social media platform $i \in \mathcal{I}$. Lemma 2 implies that at GNE, the message $m_i$ must lie in the set $\mathcal{M}_i' := \big\{m_i \in \mathcal{M}_i: \textcolor{black}{\Tilde{p}^i_l} \cdot (\Tilde{a}_l^i - a^{-i}_l) = 0, \; \forall l \in \mathcal{D}_{-i} \big\}$. For all $m_i \in \mathcal{M}_i'$, we can write the utility $u_i(m)$ as
    \begin{align} \label{thm_5_1}
    u_i(m) = \; &v_i(\textcolor{black}{\alpha}_k(m): k \in \mathcal{C}_i) + \textcolor{black}{\Tilde{p}_0} \cdot \textcolor{black}{\eta_i}(m) \nonumber \\
    &+ \sum_{l \in \mathcal{C}_{- i}} \textcolor{black}{\pi_i} ^ l \cdot \textcolor{black}{\alpha}_{i}(m) 
    - \sum_{l \in \mathcal{C}_{- i}} \textcolor{black}{\pi_l} ^ i \cdot \textcolor{black}{\alpha}_{l}(m),
    \end{align}
    where the prices $\textcolor{black}{\Tilde{p}_0}$, $\textcolor{black}{\pi_i} ^ l$, and $\textcolor{black}{\pi_l} ^ i$ for any $l \in \mathcal{C}_{-i}$ are independent of message $m_i$. We observe that $u_i(m) = u_i(\textcolor{black}{\eta_i}, \textcolor{black}{\alpha}_k : \textcolor{black}{\alpha}_k \in \mathcal{D}_{i})$. Lemma 4 implies that given a message profile $m_{-i}$ of all platforms and the government in $\mathcal{J}_{-i}$, platform $i$ can unilaterally deviate in their message $m_i \in \mathcal{M}_i$ to receive any allocation $\textcolor{black}{\alpha}_k(m) \in \mathcal{A}$, for all $k \in \mathcal{D}_{i}$. Thus, instead of the message $m_i$, we equivalently consider that the action of platform $i$ is to select the tuple $\textcolor{black}{\beta}_i = \big(\textcolor{black}{\eta_i}, \textcolor{black}{\alpha}_k : k \in \mathcal{D}_{i} \big)$, that takes values in the set $\textcolor{black}{\mathcal{B}}_i = \big\{[0,1] \times \mathcal{A}^{|\mathcal{D}_i|} : n_i \cdot h_i(\textcolor{black}{\alpha}_i) - \textcolor{black}{\eta_i} \geq 0\big\}$. For the differentiable function $h_i(a_i)$, the set $\textcolor{black}{\mathcal{B}}_i$ is convex, compact, and independent of the message profile $m_{-i}$. Similarly, the action of the government $\textcolor{black}{\alpha}_0$ takes values in the set $\mathcal{A}$ that is compact, convex, and independent of the message profile $m_{-0}$.
    
    Let the valuation $v_i\big(a_k:k\in\mathcal{C}_i\big)$ for every platform $i \in \mathcal{I}$ be quasi-concave and differentiable, and let \textcolor{black}{$\beta = \big(\beta_0, \beta_1, \dots, \beta_{|\mathcal{I}|}\big)$}. Then, for every $i \in \mathcal{I}$, the utility $u_i(\textcolor{black}{\beta})$ in \eqref{thm_5_1} is also quasi-concave and differentiable with respect to the action $\textcolor{black}{\beta}_i \in \textcolor{black}{\mathcal{B}}_i$. A similar argument implies that the government's utility $u_0(\textcolor{black}{\alpha}_0)$ is quasi concave and differentiable with respect to their action $\textcolor{black}{\alpha}_0$. Hence, it follows from Glicksberg's theorem that there exists a Nash Equilibrium (NE) for the induced game \cite{fudenberg1991}; since, by definition, any NE is also a GNE, it follows that there exists a GNE for the induced game.
    
    \textit{2) Pareto efficiency:} It is sufficient in our case to show that the NE can be characterized by a Walrasian equilibrium as all Walrasian equilibria are Pareto efficient \cite{mas_colell1995}.
    So, as in part 1, consider an arbitrary NE action profile $\textcolor{black}{\beta}^* = \big(\textcolor{black}{\alpha}_0^*, \textcolor{black}{\beta}_1^*, \dots, \textcolor{black}{\beta}_{|\mathcal{I}|}^*\big)$ that takes values in the set $\mathcal{A} \times \textcolor{black}{\mathcal{B}}_1 \times \cdots \times \textcolor{black}{\mathcal{B}}_{|\mathcal{I}|}$.
    From the definition of the NE, for every platform $i \in \mathcal{I}$ it holds that
    \begin{equation} \label{eqn:thm_NE}
        u_i(\textcolor{black}{\beta}^*) \geq u_i(\textcolor{black}{\beta}_i, \textcolor{black}{\beta}_{-i}^*), \quad \forall \textcolor{black}{\beta}_i \in \textcolor{black}{\mathcal{B}}_i.
    \end{equation}
    {Note that the NE prices $\textcolor{black}{\Tilde{p}_0}^*$, $\textcolor{black}{\pi^{*i}_l}$, $\textcolor{black}{\pi_l}^{*i}$, for all $l \in \mathcal{I}_{-i}$ cannot be influenced by platform $i$, i.e., every social media platform is a price taker. Then, using the definition of the NE in \eqref{eqn:thm_NE} with the utility $u_i(m)$ in \eqref{thm_5_1}, we can write for platform $i$ that
    \begin{multline}
        \textcolor{black}{\beta}_i^* = \arg \max_{\textcolor{black}{\beta}_i \in \textcolor{black}{\mathcal{B}}_i} \Bigg\{ v_i(\textcolor{black}{\alpha}_k: k \in \mathcal{C}_i) + \textcolor{black}{\Tilde{p}_0}^* \cdot \textcolor{black}{\eta_i} \\
        + \sum_{l \in \mathcal{I}_{- i}} \textcolor{black}{\pi_i} ^ {*l} \cdot \textcolor{black}{\alpha}_{i} - \sum_{l \in \mathcal{I}_{- i}} \textcolor{black}{\pi_l} ^ {*i} \cdot \textcolor{black}{\alpha}_{l} \Bigg\}.
    \end{multline}
    Similarly, the government also behaves as a price taker because it cannot influence the NE price $\pi^*_{0}$. For the government at NE, we can write that
    \begin{align}
        \textcolor{black}{\alpha}_0^* = \arg \max_{\textcolor{black}{\alpha}_0 \in \mathcal{A}} \{v_0(\textcolor{black}{\alpha}_0) - \pi^*_{0} \cdot \textcolor{black}{\alpha}_0\}.
    \end{align}
    It follows immediately that the NE action profile $\textcolor{black}{\beta}^*$ constitutes a Walrasian equilibrium and thus, the NE for the induced game forms a Pareto efficient equilibrium \cite{mas_colell1995}. Since any NE is also a GNE by definition, it follows that every GNE of the induced game is Pareto efficient.}
\end{proof}

\color{black}
\begin{remark}
    The GNE induced by our mechanism may not lead to allocated filters for platforms and lower bound for the government that form an optimal solution of Problem 1 using quasi-concave valuations. However, Theorem \ref{thm:quasi} establishes that a GNE still exists for such a system, and that it leads to a Pareto efficient allocation, where no player's utility can be improved without decreasing the utility of another player. Thus, from Theorem \ref{thm:ir}, we can conclude that for quasi-concave valuation functions, our mechanism incentivizes some misinformation filtering but may lead to suboptimal social welfare.
\end{remark}

\section{Discussion}\label{section:discussion}

\subsection{Interpretation of the Results}

In this subsection, we present an explanation of the mechanism presented in Section III and the main results derived in Section IV. The social planner seeks to design an efficient mechanism with the following two properties: (i) it should induce voluntary participation among all social media platforms, and (ii) it should maximize the social welfare, i.e., maximize the sum of utilities of all players. 
Note that the social welfare increases as the valuation function $v_0(a_0)$ increases, which, in turn, increases with respect to the lower bound on aggregate average trust, $a_0$. A sufficiently high lower bound $a_0$ indirectly ensures that some platforms implement non-zero filters to raise the average trust of their users. Thus, a mechanism that satisfies properties (i) and (ii) also incentivizes platforms to implement filtering, conditional on the government's valuation $v_0(a_0)$ and budget $b_0$ being sufficiently large. The challenge faced by the social planner is to achieve these properties without knowledge of the valuation function $v_0(a_0)$ of the government, the valuation function $v_i\big(a_k:k\in\mathcal{C}_i\big)$ of any platform $i \in \mathcal{I}$, and the average trust function $h_i(a_i)$ of any social media platform $i \in \mathcal{I}$.

To meet this challenge, we present a two-step mechanism in Section III. In the step one (the participation step) of the mechanism, the social planner asks each social media platform to decide whether they wish to participate in the mechanism. This is an essential question because the government is not dictatorial, i.e., it cannot force platforms to participate in the mechanism. By refusing to participate in the mechanism, platform $i$ can select no filter and pay no tax. However, platform $i$ also receives no subsidy from the government, nor benefits from the filters of platforms that do participate. We prove in Theorem \ref{thm:ir} of Section IV that the utility of any platform $i \in \mathcal{I}$ after participating in the mechanism is greater than or equal to their utility when they do not participate. Thus, the weakly dominant action of every platform in step one is to participate in the mechanism, establishing property (i).

In the step two (the bargaining step) of the mechanism, the social planner asks each player $i \in \mathcal{J}$ to broadcast a message $m_i \in \mathcal{M}_i$. Based on the message profile $m = \big(m_0,m_1,\dots,m_{|\mathcal{I}|}\big)$, the social planner allocates a minimum average trust $\eta_i(m)$, a filter $\alpha_i(m)$, and a tax $\textcolor{black}{\tau}_i(m)$ to each platform $i \in \mathcal{I}$. Similarly, she allocates a lower bound $\alpha_0(m)$ and tax $\textcolor{black}{\tau}_0(m)$ to the government. By participating in the mechanism in step one, each player $i \in \mathcal{J}$ agrees to implement the allocated filters, and either pay or receive the allocated tax. The rules defined by the social planner induce a game among the players whose equilibrium is defined as a GNE. The structure of the messages, and various parameters allocated by the social planner lead to the properties of the mechanism in Section IV.

We derive most of the properties of the mechanism in Section IV for a state where the platforms and the government are at a GNE. Lemmas \ref{lemma:truthful_prices} and \ref{lemma:truthful_action_proposals} establish preliminary properties of the tax functions $\textcolor{black}{\tau}_i(m)$ of each player $i \in \mathcal{I}$. They show that at the GNE, each player $i$ has to be consistent in their message $m_i$ with respect to the messages of other players. This consistency check ensures that no player can benefit from a manipulation of the mechanism by proposing arbitrary prices, filters, or lower bounds.
Then, we use the results of Lemmas \ref{lemma:truthful_prices} and \ref{lemma:truthful_action_proposals} to derive Theorem \ref{budget_balance}, which proves that at any GNE the mechanism is budget balanced, i.e., the sum of all taxes is $0$. This is a desirable property for the mechanism because the social planner is now guaranteed to simply take the investment of the government $\textcolor{black}{\tau}_0(m)$ and redistribute it among the social media platforms, without worrying about leftover funds or insufficient funds.
Next, we show in Lemma \ref{lemma:feasibility} that every GNE of the induced game is a feasible solution to the problem of maximizing social welfare. Lemma \ref{lemma:achievability} proves that any social media platform $i \in \mathcal{I}$ can always achieve any desired filter in $\mathcal{A}$, including $0$, by selecting an appropriate message $m_i \in \mathcal{M}_i$. This property holds irrespective of the messages selected by the other players in $\mathcal{J}_{-i}$, and establishes that a participating platform has a free choice to control their allocated filter.

All preceding results allow us to prove in Theorem \ref{thm:implementation} that every GNE of the induced game maximizes the social welfare of the system. In Theorem \ref{thm:gne_existence} we prove that the induced game is guaranteed to have at least one GNE. Theorem \ref{thm:implementation} and Theorem \ref{thm:gne_existence}, together, imply that the mechanism maximizes the social welfare of the system, establishing property (ii). Thus, we have shown that our mechanism does, indeed, incentivize platforms to filter misinformation. 

Finally, in Section IV-A we consider quasi-concave valuation functions for all players to relax some of our assumptions. In Theorem \ref{thm:quasi}, we establish that the induced game is still guaranteed to have a GNE, and that it is Pareto efficient. Thus, we observe that our mechanism still incentivizes some amount of filtering, but may lead to suboptimal social welfare.

\subsection{An Example}

In this subsection, we present a descriptive example of how our proposed mechanism may play out in a realistic setting. Consider three major social media platforms: Facebook, Twitter, and Reddit. These platforms allow users from different socioeconomic and political backgrounds to obtain the latest news. Typically, users access either Facebook, Twitter, or Reddit via their smartphone app and engage with them by scrolling down, liking, or sharing posts that feature news and personal opinions. The amount of time spent by all users on the platform and the number of actions taken by them collectively define the engagement generated by the platform \cite{allcott2017social,jaakonmaki2017}.

As user engagement is a primary driver of advertisement revenue, Facebook, Twitter, and Reddit regularly optimize their post recommendation algorithms to maximizing user engagement. Over time, these algorithms have evolved to promote posts with a high chance of generating engagement among users, without accounting for their impact on the opinions of the users \cite{tufekci2018youtube}. This has led to the formation of echo chambers, or opinion bubbles among many users, where they repeatedly interact only with posts that align with their own biases on any topic. For many users, their prior biases lead to a repeated exposure to misinformation and conspiracy theories \cite{margetts2018}. This causes
uncertainty among them regarding the integrity democratic institutions \cite{bessi2015, brown2018,tucker2017, sternisko2020dark}. For example, misinformation during elections reduces people's faith in the fairness of the election results \cite{farrell2018common}, and misinformation about precautions during a pandemic reduces people's trust in public health experts \cite{motta2020right}. 

A democratic government can observe the trust of the country's citizens from the opinions expressed by them on various social media platforms. When the government realizes the impact of misinformation on the trust of the citizens, they seek to implement policies to minimize the spread of misinformation.
In practice, each social media platform can filter misinformation by either flagging posts with inaccurate information, following them up with truthful posts, or simply not recommending them to users.
However, filtering misinformation is an expensive undertaking for platforms because of (i) the high investment required to identify inaccurate information \cite{graves2018understanding}, and (ii) potential decrease in engagement of users who are censored \cite{candogan2020optimal}. Thus, the government decides to allocate a fixed budget for the problem, and appoints an independent agency to design appropriate incentives for Facebook, Twitter, and Reddit, while staying within the budget.

The agency presents the rules of our mechanism to the government, and confirms the government's participation. Then, the agency reveals the rules of the mechanism to the platforms, and announces that platforms who choose not to participate in this collaborative effort will be labelled as non-cooperative. Furthermore, the agency assures the three platforms that they need not reveal private information and that they can choose to avoid filtering misinformation even after participating in the mechanism (Lemma \ref{lemma:achievability}). These factors ensure that each platform participates voluntarily in the mechanism (Theorem \ref{thm:ir}).
Then, the agency asks each of Facebook, Twitter, and Reddit to propose a minimum threshold to which they will raise the trust of their users in democratic institutions. The tax incentives given to each platform will be proportional to this threshold. Simultaneously, the agency asks the government to propose a minimum acceptable level for the average of all platforms' thresholds. The government's investment will be proportional to this minimum average.
The agency also asks each platform to propose various filters and prices they are willing to pay or receive for the proposed filters. Similarly, the government proposes a price for their proposed minimum average.

The agency then publicly reveals all proposals and transparently uses the rules of the mechanism to assign a potential subsidy/payment, and potential filter to each platform. Similarly, she assigns a potential amount of investment and minimum average to the government. These assignments become binding only if all stakeholders, Facebook, Twitter, Reddit, and the government, accept the assignments. If any stakeholder is dissatisfied, the agency asks all of them to change their proposals and resubmit. This process is repeated until all the stakeholders reach a consensus.
The mechanism ensures that such a consensus exists (Theorem \ref{thm:gne_existence}) and that it is the best possible result for all stakeholders (Theorem \ref{thm:implementation}). 
As long as the government is sufficiently committed to addressing the problem of misinformation, the mechanism ensures that at the consensus, the platforms will agree to implement misinformation filters. The allocations become binding on all stakeholders, and the independent agency collects the government's investment. This investment is paid out to each of Facebook, Twitter, and Reddit as a subsidy, only after they achieve the binding level of filtering.

\section{Conclusions and Future Work}\label{section:conclusion}



Our primary goal in this paper was to design a mechanism to induce a GNE solution in the misinformation filtering game, where (i) each platform agrees to participate voluntarily, and (ii) the collective utility of the government and the platforms is maximized. We designed a mechanism and proved that it satisfies these properties along with budget balance. We also presented an extension of the mechanism with weaker technical assumptions. 

Ongoing work focuses on improving the valuation and average trust functions of the social media platforms based on data. We also consider incorporating uncertainty in a platform's estimates of the impact of their filter. These refinements of the modeling framework will allow us to make our mechanism more practical for use in the real world.

Future research should include extending the results of this paper to a dynamic setting in which the social media platforms react in real-time to the proposed taxes/subsidies. In particular, someone could develop an algorithm that the players can use to iteratively arrive at the Nash equilibrium. In such an algorithm, the social planner can receive additional information from the players while they iteratively learn the GNE. Then, she can use this information to change her allocations dynamically, allowing us to relax either Assumption \ref{assumption:assumed_knowledge_b} on monitoring of average trust, or Assumption \ref{assumption:excludability} on the excludability of the platforms.

\color{black}

\bibliographystyle{IEEEtran}
\bibliography{references}

\section*{Appendix A}
In this appendix, we present an extension of by relaxing Assumption \ref{assumption:cardinality} to a more general assumption that no platform has a monopoly on its users. The mechanism presented in this assumes that for any platform $i \in \mathcal{I}$ with the set of competing platforms $\mathcal{C}_i$, it holds that $|\mathcal{C}_i| \geq 2$.

We consider the same step one (the participation step) for the mechanism as before. Then in step two (the bargaining step), the message of platform $i$ is defined as
\begin{equation}\label{eqn:defn_message_2}
    m_i := (\textcolor{black}{\Tilde{h}_i}, \textcolor{black}{\Tilde{p}_i}, \Tilde{a}_i),
\end{equation}
where
$\textcolor{black}{\Tilde{h}_i} \in \mathbb{R}_{\geq 0}$ is the minimum average trust that platform $i$ proposes to achieve through filtering; $\textcolor{black}{\Tilde{p}_i}$, is the collection of prices that platform $i$ is willing to pay or receive per unit changes in the filters of other competing platforms \textcolor{black}{(except $i$)} and \textcolor{black}{the government's lower bound}, given by
\begin{equation}\label{eqn:prices_2}
    \textcolor{black}{\Tilde{p}_i} := 
    \begin{aligned}
    \begin{cases}
    (\textcolor{black}{\Tilde{p}_l} ^ i : l \in \mathcal{D}_{i}), \quad &\text{if } |\mathcal{C}_{i}| = 2, \\
    (\textcolor{black}{\Tilde{p}_l} ^ i : l \in \mathcal{D}_{-i}), \quad &\text{if } |\mathcal{C}_{i}| \geq 3,
    \end{cases}
    \end{aligned}
\end{equation}
where $\textcolor{black}{\Tilde{p}_l} ^ i \in \mathbb{R}_{\geq0}$ for all $i, l \in \mathcal{J}$;
and $\Tilde{a}_i := (\Tilde{a}_k ^ i: k \in \mathcal{D}_i),$ with $\Tilde{a}_i \in \mathbb{R} ^ {|\mathcal{D}_{i}|}$ is the profile of filters for all competing platforms \textcolor{black}{(including $i$)} and \textcolor{black}{government's lower bound proposed by platform $i$.}

The message of the government is  $m_0 := (\textcolor{black}{\Tilde{p}_0}, \Tilde{a}_0 ^ 0)$, where $\textcolor{black}{\Tilde{p}_0} \in \mathbb{R}_{\geq 0}$ is the price that the government is willing to pay or receive per unit change of the average trust, and $\Tilde{a}_0^0 \in \mathbb{R}$ is the \textcolor{black}{lower bound} proposed by the government.

\textcolor{black}{Based on the message profile $m := (m_0, m_1,$ $\dots, m_{|\mathcal{I}|})$ that the social planner receives, she allocates the following parameters to the players:}

\textcolor{black}{\textit{1)} The social planner allocates a filter to each platform $i \in \mathcal{I}$ and a lower bound to the government such that the constraints of Problem $1$ are satisfied.}
The filter allocated by the social planner to platform $i$ is $\textcolor{black}{\alpha}_i(m) := \sum_{k \in \mathcal{C}_i} \frac{\Tilde{a}_i ^ k}{|\mathcal{C}_i|}$. The lower bound allocated by the social planner to the government is $\textcolor{black}{\alpha}_0(m) := \sum_{k \in \mathcal{J}} \frac{\Tilde{a}_0 ^ k}{|\mathcal{J}|}$. 

\textcolor{black}{\textit{2)}} The social planner allocates a minimum average trust {$\textcolor{black}{\eta_i}(m) \in [0,1]$} to each platform $i \in \mathcal{I}$, given by
\begin{equation}\label{eqn:defn_of_eta_2}
    \textcolor{black}{\eta_i}(m) := \min\left\{ \frac{n_i \cdot \textcolor{black}{\Tilde{h}_i}}{\sum_{k \in \mathcal{I}} n_k \cdot \textcolor{black}{\Tilde{h}_k}} \cdot \textcolor{black}{\alpha}_0(m), \; 1\right\},
\end{equation}
where the social planner will not accept a message $m_i$ from a platform $i$ that might lead to a situation where $\sum_{k \in \mathcal{I}} n_k \cdot \textcolor{black}{\Tilde{h}_k} = 0$.
\textcolor{black}{The allocated minimum average trust, $\eta_i(m)$, is a lower bound on average trust that must be achieved by platform $i$.}
Let the filter implemented by platform $i$ be $a_i$. Then, platform $i$ must ensure that $n_i \cdot h_i(a_i)\geq \textcolor{black}{\eta_i}(m)$.
Recall from Section III-B that, as a result of Assumption \ref{assumption:assumed_knowledge_b}, the social planner can prevent the platforms from violating the constraint imposed by $\eta_i(m)$.

\textcolor{black}{\textit{3)}}
The social planner also allocates a payment price
\begin{align} \label{eqn:tau_def_2}
    \textcolor{black}{\pi_l}  ^ i := 
    \begin{cases}
        \Tilde{p}_l^l, \quad &\text{if } |\mathcal{C}_l| = 2, \\
        \sum_{k \in \mathcal{C}_{- l} : k \neq i} \dfrac{\textcolor{black}{\Tilde{p}_l} ^ k}{|\mathcal{C}_{l}| - 2}, \quad &\text{if } |\mathcal{C}_l| \geq 3, 
    \end{cases}
\end{align}
where $\textcolor{black}{\pi^i_l} \in \mathbb{R}_{\geq 0},$ to be paid by platform $i \in \mathcal{I}$ for a unit change in allocated filter $\textcolor{black}{\alpha}_l(m)$ of every other competing platform $l \in \mathcal{C}_{- i}$. Furthermore, the social planner allocates a subsidy price
\begin{align} \label{eqn:sigma_def_2}
    \sigma^i_l := 
    \begin{cases}
        \Tilde{p}^l_i, \quad &\text{if } |\mathcal{C}_l| = 2, \\
        \sum_{k \in \mathcal{C}_{- l} : k \neq i} \dfrac{\textcolor{black}{\Tilde{p}_l} ^ k}{|\mathcal{C}_{l}| - 2}, \quad &\text{if } |\mathcal{C}_l| \geq 3, 
    \end{cases}
\end{align}
where $\sigma^i_l \in \mathbb{R}_{\geq 0},$ to be received by platform $i \in \mathcal{I}$ from every other competing platform $l \in \mathcal{C}_{- i}$, for a unit change in allocated filter $\alpha_i(m)$. For the government, the social planner simply allocates a price $\pi_0 := \sum_{i \in \mathcal{I}} \frac{\textcolor{black}{\Tilde{p}^i_0}}{|\mathcal{I}|}$ to be paid for a unit change in lower bound $\alpha_0(m)$.

\begin{remark}
    Note that when $|\mathcal{C}|_i = 2$, platform $i \in \mathcal{I}$ proposes a price corresponding to their own proposed action $\Tilde{a}_i(m)$. In contrast, when $|\mathcal{C}_i| \geq 3$, platform $i$ does not proposes a price corresponding to their own filter. However, we have designed the payment price in \eqref{eqn:tau_def_2} and subsidy price \eqref{eqn:sigma_def_2} so that platform $i$ cannot affect either of these prices with their message $m_i$. Thus, each still platform behaves as a \textit{price taker} when $|\mathcal{C}_i| = 2$.
\end{remark}

\textcolor{black}{\textit{4)}} The social planner \textcolor{black}{allocates} the following tax to each social media platform $i \in \mathcal{I}$,
\begin{gather}
    \textcolor{black}{\tau}_i := - \textcolor{black}{\Tilde{p}_0} \cdot \textcolor{black}{\eta_i}(m) - \sum_{l \in \mathcal{C}_{- i}} \sigma^i _ l \cdot \textcolor{black}{\alpha}_{i}(m) 
    + \sum_{l \in \mathcal{C}_{- i}} \textcolor{black}{\pi^i_l} \cdot \textcolor{black}{\alpha}_{l}(m) \nonumber \\
    + \sum_{l \in \mathcal{C}_{- i} \cup \{0\}} \textcolor{black}{\Tilde{p}_l} ^ i \cdot (\Tilde{a}_l ^ i - \Tilde{a}_l ^ {- i}) ^ 2 \nonumber \\
    + \sum_{l \in \mathcal{C}_{-i}} \Big( \mathbb{I}(|\mathcal{C}_{i}| = 2)\cdot (\Tilde{p}_i^i - \Tilde{p}^l_i)^2 + \mathbb{I}(|\mathcal{C}_{l}| = 2) \cdot (\Tilde{p}_l^i - \Tilde{p}_l^l)^2 \Big), \label{eqn:payment_function_2}
\end{gather}
where $\mathbb{I}(\cdot)$ is the indicator function, $\Tilde{a}_l ^ {- i} = \sum_{k \in \mathcal{C}_{- l}} \frac{\Tilde{a}_l ^ k}{|\mathcal{C}_{l}| - 1}$, for each $l \in \mathcal{C}_{- i}$,
is the average of the proposed filters for $l$ by all competing platforms except $i\in\mathcal{I}$, and and $\Tilde{a}_0 ^ {- i} = \sum_{k \in \mathcal{J}_{-i}} \frac{\Tilde{a}_0 ^ k}{|\mathcal{J}| - 1}$ is the average of lower bounds proposed by all players except $i$. The tax $\textcolor{black}{\tau}_i(m)$ of platform $i$ in \eqref{eqn:payment_function_2} can be interpreted as follows: (i) the first term in \eqref{eqn:payment_function_2}  represents a subsidy given by the government to platform $i$ for the increase in average trust among the users of platform $i$; (ii) the second term in \eqref{eqn:payment_function_2} is a collection of subsidies given by each competing platform $l \in \mathcal{C}_{- i}$ to platform $i$ for the increase in valuation $v_l\big(\alpha_k(m): k \in \mathcal{C}_l\big)$ due to the allocated filter $\textcolor{black}{\alpha}_i(m)$; (iii) the third term in \eqref{eqn:payment_function_2} is a payment by platform $i$ for the increase in valuation $v_i\big(\alpha_k(m):k \in \mathcal{C}_i\big)$ due to the allocated filter $\alpha_l(m)$ of each competing platform $l \in \mathcal{C}_{- i}$; (iv) the fourth term in \eqref{eqn:payment_function_2} is a collections of penalties to platform $i$ if either the filter proposed in message $m_i$ for any competing platform $l \in \mathcal{C}_{- i}$ is inconsistent the filters proposed by other platforms, or if the lower bound proposed in $m_i$ is inconsistent with the lower bound proposed by other players; and (v) the fifth term is a collection of penalties to social media $i$ for inconsistency in the proposed price, only applicable if $|\mathcal{C}_{i}| = 2$, or if $|\mathcal{C}_{l}| = 2$ for some $l \in \mathcal{C}_{-i}$.

The social planner also proposes the following payment function to the government:
\begin{equation}\label{eqn:payment_function_gov_2}
    \textcolor{black}{\tau}_0 := \pi_0 \cdot \textcolor{black}{\alpha}_0(m) + (\textcolor{black}{\Tilde{p}_0} - \pi_0) ^ 2,
\end{equation}
where the first term is the total investment made by the government for the allocated lower bound $\textcolor{black}{\alpha}_0(m)$, and the second term is a penalty when the price proposed by the government deviates from the price allocated to the government.

\begin{remark}
    Note that the presence of the indicator function $\mathbb{I}(\cdot)$ in \eqref{eqn:payment_function_2} does not lead to discontinuities in the utility $u_i(m_i,m_{-i})$ with respect to the message profile $m$.
\end{remark}

\begin{remark}
    The extended mechanism induces a game where the strategy of platform $i \in \mathcal{I}$ is $m_i \in \mathcal{M}_i$, such that $\alpha_i(m) \in \mathcal{S}_i(m)$. The equilibrium for the induced game is given by the GNE, defined in \eqref{eqn:GNE_defn}.
\end{remark}

Then, we note that the results of Lemmas \ref{lemma:truthful_prices} - \ref{lemma:truthful_action_proposals} hold for the extended mechanism. In addition, we prove in Lemma \ref{lemma:extended} that the equilibrium price received by any platform $i \in \mathcal{I}$ for an allocated filter $\alpha_l(m)$, $l \in \mathcal{C}_{-i}$, is the same as the price paid by the competing platform $l$.

\begin{lemma} \label{lemma:extended}
    Let message profile $m^* \in \mathcal{M}$ be a GNE of the induced game. Then, for each platform $i \in \mathcal{I}$ and each competing platform $l \in \mathcal{C}_{-i}$, it holds that $\sigma_l^{*i} = \textcolor{black}{\pi_i}^{*l}$.
\end{lemma}

\begin{proof}
    Consider two social media platforms $i \in \mathcal{I}$ and $l \in \mathcal{C}_{-i}$. The result holds from the definition of $\sigma_l^i$ and $\textcolor{black}{\pi_i}^l$ when both $|\mathcal{C}_{i}| \geq 3$ and $|\mathcal{C}_{l}| \geq 3$.
    
    Let $|\mathcal{C}_{i}| = 2$, with $\mathcal{C}_i = \{i,l\}$, and let $m_{-i}^*$ be the message profile at GNE of all players except $i$. In order to maximize their utility $u_i(m_i,m_{-i}^*)$, platform $i$ must select a price $\Tilde{p}_i^{*i}$ that minimizes the tax $\textcolor{black}{\tau}_i$ in \eqref{eqn:payment_function_2}. Thus, $\frac{\partial u_i}{\partial \Tilde{p}^i_i} \big|_{\Tilde{p}^{*i}_i} = 2 \cdot (\Tilde{p}^{*i}_i - \Tilde{p}^{*l}_i) = 0,$ which yields $\Tilde{p}^{*i}_i = \Tilde{p}^{*l}_i$. Then, the result holds using the definitions of $\sigma_l^i$ and $\textcolor{black}{\pi_i}^l$ in \eqref{eqn:sigma_def_2} and \eqref{eqn:tau_def_2}, respectively. Through a similar analysis, we can prove the result when $|\mathcal{C}_{l}| = 2$.
\end{proof}

An additional implication of Lemma \ref{lemma:extended} is that the fifth term in \eqref{eqn:payment_function_2} is $0$ at any GNE. Thus, it can be verified that the results of Lemmas \ref{lemma:feasibility} - \ref{lemma:achievability} and Theorems \ref{budget_balance} - \ref{thm:quasi} hold for the extended mechanism.

\end{document}